\documentclass[11pt]{article}

\usepackage[font=small,labelfont=bf]{caption}

\usepackage{enumitem}
\usepackage{graphicx}
\usepackage{xcolor}
\usepackage{blkarray}
\usepackage{amsthm}
\usepackage{amsmath,amssymb}
\usepackage{bbm}
\usepackage{mathtools}
\usepackage{pdflscape}
\usepackage{float}
\usepackage{rotating}
\usepackage{titlesec}
\usepackage{tcolorbox}

\usepackage[
citestyle=numeric,
style=numeric,
uniquename=false,
sorting=ynt,
maxbibnames=8
]{biblatex}

\newtheorem{theorem}{Theorem}[section]

\newtheorem{assumption}{Assumption}
\newtheorem{claim}{Claim}

\DeclareUnicodeCharacter{0301}{\'{e}} 

\titleformat{\paragraph}
{\normalfont\normalsize\bfseries}{\theparagraph}{1em}{}
\titlespacing*{\paragraph}
{0pt}{3.25ex plus 1ex minus .2ex}{1.5ex plus .2ex}

\DeclareMathOperator{\EX}{\mathbbm{E}}
\allowdisplaybreaks

\begin{document}
\setlength{\abovedisplayskip}{4pt}
\setlength{\belowdisplayskip}{10pt}
\setlength{\abovedisplayshortskip}{4pt}
\setlength{\belowdisplayshortskip}{10pt}

\title{Superinfection and the hypnozoite reservoir for \textit{Plasmodium vivax}: a general framework}

\author{Somya Mehra$^{1}$ \and James M. McCaw$^{1,2}$ \and \and Peter G. Taylor$^1$}

\date{%
    $^1$School of Mathematics and Statistics, The University of Melbourne, Parkville,  Australia\\%
    $^2$Centre for Epidemiology and Biostatistics, Melbourne School of Population and Global Health, The University of Melbourne, Parkville, Australia
}

\maketitle

\section*{Abstract}

Malaria is a parasitic disease, transmitted by mosquito vectors. \textit{Plasmodium vivax} presents particular challenges for disease control, in light of an undetectable reservoir of latent parasites (hypnozoites) within the host liver. The burden of blood-stage infection for \textit{P. vivax} is two-fold, driven by both primary infections (resulting directly from mosquito bites) and relapses (resulting from hypnozoite activation). Superinfection, which is driven by temporally proximate mosquito inoculation and/or hypnozoite activation events, is an important feature of \textit{P. vivax}. Here, we present a model of hypnozoite accrual and superinfection for \textit{P. vivax}. To couple host and vector dynamics for a homogeneously-mixing population, we construct a density-dependent Markov population process with countably many types, for which disease extinction is shown to occur almost surely. We also establish  a functional law of large numbers (FLLN), taking the form of an infinite-dimensional system of ordinary differential equations (ODEs) that can also be recovered by coupling expected host and vector dynamics (i.e. a hybrid approximation) or through a standard compartment modelling approach.\\

Recognising that the subset of these equations that model the infection status of the human hosts has precisely the same form as the Kolmogorov forward equations for a Markovian network of infinite server queues with an inhomogeneous batch arrival process, we use physical insight into the evolution of the latter process to write down a time-dependent multivariate generating function for the solution. We use this characterisation to collapse the infinite-compartment model into a single integrodifferential equation (IDE) governing the intensity of mosquito-to-human transmission. Through a steady state analysis, we recover a threshold phenomenon for this IDE in terms of a parameter $R_0$ expressible in terms of the primitives of the model, with the disease-free equilibrium shown to be uniformly asymptotically stable if $R_0<1$ and an endemic equilibrium solution emerging if $R_0>1$. Our work provides a theoretical basis to explore the epidemiology of \textit{P. vivax}, and introduces a strategy for constructing tractable population-level models of malarial superinfection that can be generalised to allow for greater biological realism in a number of directions.

\section{Introduction}
Malaria is a parasitic, vector-borne disease with a staggering pulic health burden. The overwhelming majority of malaria cases (98\%) are attributed to the parasite \textit{Plasmodium falciparum}, particularly in Africa \parencite{who2021}. \textit{Plasmodium vivax}, however, exhibits a broader geographical distribution, driving much of the malaria burden in South East Asia, the Americas, the Western Pacific and the Eastern Mediterranean \parencite{who2021}.  While an estimated 4.5 million malaria cases were attributed to \textit{P. vivax} in 2020 alone \parencite{who2021}, morbidity arising from \textit{P. vivax} infections remains ``obscure and insidious'' \parencite{battle2021global}.\\

The transmission of malaria parasites to humans is mediated by \textit{Anopheles} mosquito vectors. During the course of a bloodmeal, an infected mosquito can transmit parasites (sporozoites) to a human host. Following a period of liver-stage development (exoerythrocytic schizogony), parasites are released into the bloodstream, giving rise to a blood-stage infection that is sustained by the replication of parasites in invaded red blood cells \parencite{venugopal2020plasmodium}. A key epidemiological characteristic of malaria is the phenomenon of superinfection. Since the circulation and replication of (pre-existing) parasites in the bloodstream does \textit{not} preclude further blood-stage infection, an individual can concurrently harbour multiple co-circulating broods of blood-stage parasites. In the context of \textit{P. falciparum}, we define each blood-stage `brood' to derive from a single infective bite. The interpretation of a blood-stage `brood' for \textit{P. vivax} is more nuanced, in light of its ability to cause relapsing infections following the accrual of a ``hypnozoite reservoir'' \parencite{white2012relapse, white2014modelling}. Notably, a \textit{P. vivax} parasite (sporozoite) injected into a human host has two possible fates: it either gives rise to a primary (blood-stage) infection within approximately 9 days of the bite itself \parencite{mikolajczak2015plasmodium}, or develops into a hypnozoite \parencite{mueller2009key}. Hypnozoites undergo indeterminate latency periods, often lasting weeks or months, during which they are undetectable using standard techniques \parencite{schafer2021plasmodium}. The activation of each hypnozoite, however, can trigger an additional blood-stage infection, known as a relapse \parencite{mueller2009key}. For \textit{P. vivax}, each primary infection and relapse is defined to comprise a separate blood-stage brood.\\

We define the multiplicity of broods (MOB) to be the number of co-circulating blood-stage broods in a host at a given point in time, with superinfection taken to be a collective term for blood-stage infections with MOB$>1$. As such, superinfection arises from temporally proximate reinfection (that is, infective bites) and, in the case of \textit{P. vivax}, hypnozoite activation events. With epidemiological data indicating the preponderance of relapses over primary infections \parencite{commons2020estimating} and evidence of \textit{P. vivax} superinfection even in the absence of reinfection \parencite{popovici2018genomic}, analysis of the statistics of superinfection for \textit{P. vivax} warrants careful consideration of the hypnozoite reservoir.\\


The classical framework of malarial superinfection for \textit{P. falciparum}, proposed initially by \textcite{macdonald1950analysis} and formulated mathematically by \textcite{bailey1957}, assumes independent clearance of each brood, without imposing an upper bound on the MOB. Under this setting, a natural construction to describe the within-host dynamics of superinfection is an infinite-server queue with a time dependent arrival rate given by the intensity of mosquito-to-human transmission \parencite{dietz1974malaria, nedelman1984inoculation, smith2009endemicity, henry2020hybrid}. Hereafter, we refer to the mosquito-to-human transmission intensity, as quantified by the infective bite rate per human, as the force of reinfection (FORI). For \textit{P. vivax}, introducing the additional assumption that the dynamics of each hypnozoite are governed by independent stochastic processes \parencite{white2014modelling}, we have recently extended this idea to characterise within-host and superinfection dynamics using an {\it open network} of infinite-server queues with geometrically-distributed batch arrivals at a time-dependent rate given by the FORI \parencite{mehra2021antibody, mehra2022hypnozoite}. In \textcite{mehra2022hypnozoite}, we derive a time-dependent generating function for the state of the queueing network, which can be inverted analytically to recover marginal distributions for MOB and the hypnozoite burden, amongst other quantities of epidemiological interest, on a within-host scale.\\

Mathematical modelling of superinfection and hypnozoite dynamics at the \textit{population-level} can be challenging. As we note in \parencite{mehra2022hypnozoite}, a common approach in the construction of transmission models of \textit{P. vivax} has been to consolidate hypnozoite carriage into a single state, with an accompanying parametric form for the time to first relapse for hypnozoite-positive individuals \parencite{ishikawa2003mathematical, aguas2012modeling, roy2013potential, chamchod2013modeling, robinson2015strategies, white2016variation}. The binarisation of hypnozoite carriage, however, obscures the relationship between transmission intensity and hypnozoite accrual: there is no variation in the risk of relapse based on hypnozoite density, the limitations of which are discussed in \parencite{mehra2022hypnozoite}. Explicit variation in the hypnozoite burden is captured in the deterministic models of \textcite{white2014modelling, white2018mathematical, anwar2021multiscale}. While the `batch' model of \parencite{white2018mathematical} captures `broods' of hypnozoites in the liver, it ignores variation in parasite inoculum sizes. On the other hand, both \parencite{white2014modelling, anwar2021multiscale} take hypnozoite densities into account. To account for superinfection, \parencite{white2014modelling} employs a ``pseudoequilibrium approximation'' for the (blood-stage) infection recovery rate, adopting the functional form derived by \textcite{dietz1974malaria}. This functional form, which was derived in the absence of hypnozoite accrual, is embedded in multiple transmission models for \textit{P. falciparum} as a proxy for superinfection \parencite{hagmann2003malaria, gemperli2006malaria, smith2007standardizing, alonso2019critical}. However, we argue that this functional form, as embedded by \parencite{white2014modelling} in a model for \textit{P. vivax}, is not appropriate in a model that takes hypnozoite accrual into account (see Appendix \ref{appendix::pseudoeq} for details). The multiscale model of \parencite{anwar2021multiscale} is intended to serve as a re-formulation of the infinite-compartment model proposed in \parencite{white2014modelling}. To incorporate hypnozoite accrual in a simple population-level framework, it draws on the relapse rate, conditional on the absence of blood-stage infection, derived in \textcite{mehra2022hypnozoite}, to derive a numerically tractable system of integrodifferential equations (IDEs). However, while the above-mentioned conditional relapse rate is derived under a framework that explicitly allows for superinfection, the population-level model of \parencite{anwar2021multiscale} does not have separate compartments for different values of the MOB (see Appendix \ref{appendix::superinf_ide} for details, and a proposed correction to  \parencite{anwar2021multiscale} that has been adopted in subsequent work \parencite{anwar2023optimal}).\\

To the best of our knowledge, the thesis of \textcite{mehra2022superinf} --- which forms of the basis of this paper --- derives the first model that characterises both superinfection and hypnozoite dynamics for \textit{P. vivax}. The model assumptions, building on the work of \textcite{white2014modelling}, are detailed in Section \ref{sec::vivax_queue_network}, where we re-visit the queueing model introduced in \textcite{mehra2022hypnozoite} to characterise within-host superinfection and hypnozoite dynamics as a function of the FORI. We then construct a density-dependent Markov population process to couple host and vector dynamics in a homogeneously-mixing population, whilst allowing for superinfection and the accrual of the hypnozoite reservoir in the absence of human demographics in Section \ref{sec::vivax_markov}, proving that disease extinction occurs almost surely (Theorem \ref{proof::vivax_markov_extinction}). Using the work of \textcite{barbour2012law}, we obtain a functional law of large numbers (FLLN), that takes the form of an infinite-dimensional system of ordinary differential equations (ODEs) (Section \ref{sec::vivax_lln_result}) for which the model of \textcite{bailey1957} arises as a special case (Appendix \ref{sec::bailey}). By drawing on our previous analysis of the within-host model, we show that FLLN can be reduced to a single IDE governing the time evolution of the FORI; the dynamics of superinfection and the hypnozoite reservoir in the human population can be recovered as a function of the FORI solving this IDE (Section \ref{sec::lln_ide}). We then establish a threshold phenomenon for the reduced IDE, with the disease-free equilibrium shown to be uniformly asymptotically stable if $R_0<1$, and an endemic equilibrium solution emerging iff $R_0>1$ (Theorem \ref{theorem::ide_ss}, Section \ref{sec::lln_ss}).\\

A recurrent theme in our analysis of population-level models is the utility of a physical understanding of the within-host model, governing the hypnozoite/MOB burden within a single human as a function of the intensity of mosquito-to-human transmission. In Section \ref{sec::discussion}, we discuss how the ideas presented in this manuscript constitute a general strategy that can be employed to construct tractable population-level models of malarial superinfection, with appropriate constraints on the underlying model structure.

\section{Modelling within-host hypnozoite and superinfection dynamics using an open network of infinite server queues with batch arrivals} \label{sec::vivax_queue_network}

\textcite{white2014modelling} construct a within-host model for short-latency (tropical) strains of \textit{P. vivax}, predicated on the following set of assumptions:
\begin{itemize}
    \item Each infective mosquito bite immediately gives rise to a primary (blood-stage) infection and establishes a batch of hypnozoites in the liver.
    \item Hypnozoite batch sizes $N_i$ (for the $i^\text{th}$ bite) which are independent and identically-distributed (i.i.d.) across bites, are geometrically-distributed with probability mass function
    \begin{align}
        P( N_i = n ) = \frac{1}{1 + \nu} \Big( \frac{\nu}{1 + \nu} \Big)^n. \label{eq:bite_pmf}
    \end{align}
    for $n \in \mathbbm{Z}_{\geq 0}$ and mean $\EX[N_i]=\nu$.
    \item Each hypnozoite in the liver undergoes activation at constant rate $\alpha$, which gives rise to the host suffering a (blood-stage) relapse; and death at constant rate $\mu$.
    \item Hypnozoites behave independently; that is, the dynamics of each of the hypnozoites is described by an independent stochastic process.
\end{itemize}

We extend the framework of \textcite{white2014modelling} to explicitly allow for superinfection. Specifically, we make the assumption that each blood-stage infection (relapse or primary) is naturally cleared at constant rate $\gamma$, with independent dynamics for each hypnozoite and blood-stage infection; as such, the existence of a previous blood-stage infection does not preclude further blood- and liver-stage infections, nor alter the rate of clearance for subsequent blood-stage infections.\\

We first examine the within-host dynamics of superinfection and the hypnozoite reservoir as a function of the FORI. In \textcite{mehra2022hypnozoite}, we construct an open network of infinite server queues with batch arrivals to concurrently describe hypnozoite accrual and the burden of blood-stage infection (allowing for superinfection) as a function of mosquito-to-human transmission intensity. The formulation of \parencite{mehra2022hypnozoite} can be extended to account for long-latency phenotypes, which are predominantly found in temperate regions \parencite{white2014modelling, mehra2020activation}, and the administration of drug treatment at a pre-determined sequence of times. Here, we consider the simplest case of the model presented in \parencite{mehra2022hypnozoite}, restricting our attention to short-latency phenotypes (characteristic of tropical transmission settings) \parencite{white2014modelling} in the absence of drug treatment.\\

We begin by delineating the set of possible states that a single hypnozoite can occupy:
\begin{itemize}
    \item $H$ indicates a hypnozoite that is currently present in the liver;
    \item $A$ indicates a hypnozoite that has activated to give rise to a relapse that is currently in progress;
    \item $C$ indicates a hypnozoite that has previously given rise to a relapse, which has since been cleared;
    \item $D$ indicates a hypnozoite that has died, rather than activating.
\end{itemize}

Hypnozoites in the liver (state $H$) undergo activation at rate $\alpha$ and death at rate $\mu$. Relapses (state $A$) are cleared from the bloodstream at rate $\gamma$. This model, introduced in \parencite{mehra2022hypnozoite}, is a simple extension of the short-latency model of \parencite{white2014modelling} that accounts for the clearance of blood-stage infection. A continuous-time Markov chain model that captures the above dynamics has non-zero transition rates 
\begin{align*}
    &q(H,A) = \alpha \qquad q(H, D) = \mu \qquad q(A, C) = \gamma
\end{align*}
with absorbing states $C$ and $D$.\\

We denote by $p_s(t)$ the probability that a hypnozoite is in state $s \in \{ H, A, C, D \} := S_h$ at time $t$ after inoculation. It is straightforward to establish that when $\alpha + \mu \neq \gamma$,
\begin{align}
    p_H(t) &= e^{-(\alpha + \mu)t}\\
    p_A(t) &= \frac{\alpha}{(\alpha+\mu) - \gamma} \big( e^{-\gamma t} - e^{-(\alpha + \mu)t} \big)\\
    p_C(t) &= \frac{\alpha}{\alpha + \mu} \big( 1 - e^{-(\alpha + \mu)t} \big) - \frac{\alpha}{(\alpha+\mu) - \gamma} \big( e^{-\gamma t} - e^{-(\alpha + \mu)t} \big) \\
    p_D(t) &= \frac{\mu}{\alpha + \mu} \big( 1 - e^{-(\alpha + \mu)t} \big).
\end{align}
as in Equations (13) to (16) of \parencite{mehra2022hypnozoite}. In practice, we would expect the mean duration of hypnozoite carriage $1/(\alpha+\mu)$ to exceed the mean duration of each blood-stage infection $1/\gamma$.\\

Likewise, we delineate the state space for each primary infection:
\begin{itemize}
    \item $P$ indicates a primary infection that is currently in progress
    \item $PC$ indicates a primary infection that has been cleared from the bloodstream.
\end{itemize}

We assume that each bite necessarily triggers a primary infection that is cleared naturally from the bloodstream at the constant rate $\gamma$, yielding a continuous-time Markov chain model with transition rates
\begin{align*}
    q(P, PC) = \gamma \qquad q(PC, P) = 0.
\end{align*}

To embed our model for a single hypnozoite in an epidemiological framework, we construct an open network of infinite server queues, labelled $H, A, C, D, P, PC$ (Figure \ref{fig:vivax_queue}). The arrival process, comprising mosquito bites, is governed by non-homogeneous Poisson process with a time-dependent rate $\lambda(t)$, such that
\begin{align*}
    \int^t_0 \lambda(\tau) d \tau < \infty \text{ for all } t \geq 0.
\end{align*}

Each bite leads to the arrival of a single `individual' in queue $P$ (that is, a primary infection), in addition to a geometrically-distributed batch (with PMF (\ref{eq:bite_pmf})) in queue $H$ (representing the hypnozoite reservoir). Hypnozoites in queue $H$ follow the dynamics described above in moving to queues $A$, $C$ or $D$. A primary infection in queue $P$ moves to queue $PC$ at rate $\gamma$.\\

Denote by $N_s(t)$ the number of `individuals' (either hypnozoites or infections) in queue $s \in \{ H, A, C, D, P, PC \} := S$ at time $t$. From Equation (39) of \parencite{mehra2022hypnozoite}, given
\begin{align*}
    N_H(0) = N_A(0) = N_D(0) = N_C(0) = N_P(0) = N_{PC}(0) = 0,
\end{align*}
the joint PGF for ${\bf z} = (z_H, z_A, z_D, z_C, z_P, z_{PC})$
\begin{align}
     G(t, \mathbf{z}) & := \EX \Big[ \prod_{s \in S} z_s^{N_s(t)} \Big]  = \exp \bigg\{ - \int^t_0 \lambda(\tau) \Big[ 1 -  \frac{e^{-\gamma (t-\tau)} z_P + (1-  e^{-\gamma (t-\tau)}) z_{PC}}{1 + \nu \Big( 1- \sum_{s \in S_h} z_s \cdot p_s(t-\tau) \Big)} \Big] d \tau \bigg\} \label{vivax_multi_pgf}
\end{align}
is guaranteed to converge in the domain $\mathbf{z} \in [0, 1]^6$.\\

\begin{figure}
    \centering
    \includegraphics[width=0.7\textwidth]{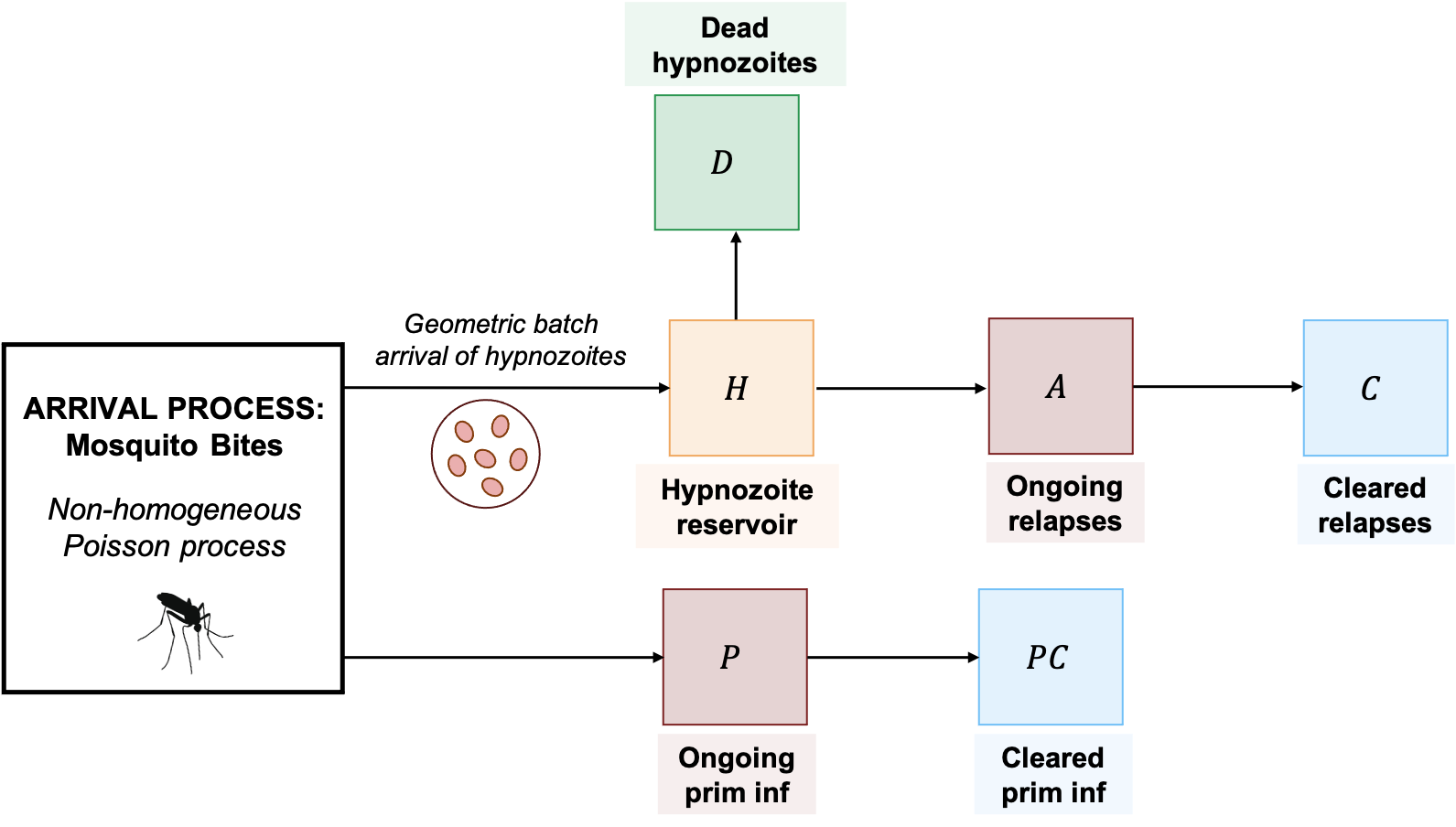}
    \caption{Schematic of the open network of infinite server queues governing the within-host hypnozoite and MOB burden, as a function of the intensity of mosquito-to-human transmission. Adapted from Figure 3 of \textcite{mehra2022hypnozoite} as a special case (short-latency hypnozoites, no drug treatment).}
    \label{fig:vivax_queue}
\end{figure}
Explicit formulae for quantities of epidemiological interest --- including marginal distributions for MOB and the hypnozoite burden; the proportion of recurrences that are expected to be relapses and the cumulative burden of blood-stage infection over time --- are provided in \parencite{mehra2022hypnozoite}. There the marginal distributions for MOB and the hypnozoite burden are expressed in terms of partial exponential Bell polynomials. We can equivalently formulate recurrence relations to compute marginal distributions of interest following the approach of \textcite{willmot2001transient}, or compute joint distributions using the recurrence relations we derive in \textcite{mehra2023open}.

\section{A Markov population process with countably many types} \label{sec::vivax_markov}
To couple vector and human dynamics, whilst accounting for superinfection and the hypnozoite reservoir in the absence of human demographics (that is, without allowing for births/deaths in the human population), we construct a density-dependent Markov population process with countably many types. We consider a closed, homogeneously mixing population, with fixed numbers of humans $P_H$ and mosquitoes $P_M$, where:
\begin{itemize}
    \item Each mosquito bites humans at constant rate $\beta$.
    \item When an uninfected mosquito bites a \textit{blood-stage} infected human (that is, a human with at least one ongoing relapse or primary infection), human-to-mosquito transmission occurs with probability $q$.
    \item When an infected mosquito bites \textit{any} human in the population, mosquito-to-human transmission occurs with probability $p$. This involves the establishment of:
    \begin{itemize}
        \item a primary (blood-stage) infection, which increases the MOB of the human host by one; and
        \item a batch of hypnozoites in the liver, where the batch sizes are geometrically-distributed with PMF (\ref{eq:bite_pmf}) and are i.i.d. across bites.
    \end{itemize}
    \item Within a human, each hypnozoite and primary infection is governed by an independent stochastic processes. 
    \item Within a human, each hypnozoite undergoes activation at rate $\alpha$ and death at rate $\mu$.
    \item Within a human, each primary infection or relapse (triggered by a hypnozoite activation event) is cleared independently at constant rate $\gamma$.
    \item Each infected mosquito is replaced by an uninfected mosquito at rate $g$.
\end{itemize}

We can formulate the state space of the Markov chain in several ways. One approach is to label each human with a hypnozoite/MOB state, and each mosquito with a binary infection state. Specifically, we can denote the state at time $t$ as $(\mathbf{u}(t), \mathbf{v}(t))$ where
\begin{itemize}
    \item $\mathbf{u}(t)$ is a $P_H$-dimensional vector whose $r^{th}$ component itself is a vector giving the number of `individuals' in compartment $H$, and the sum of individuals in compartments $A$ and $P$ in Figure \ref{fig:vivax_queue} for the $r^{th}$ human; and
    \item $\mathbf{v}(t)$ is a is a $P_M$-dimensional vector whose $s^{th}$ component is 1 if the $s^{th}$ mosquito is infected and zero otherwise,
\end{itemize}
yielding the state space
\begin{align*}
    \chi' = (\mathbbm{Z}_{\geq 0} \times \mathbbm{Z}_{\geq 0})^{P_H} \times \{0, 1\}^{P_M}.
\end{align*}

Alternatively, we can formulate the state space to count the number of humans and mosquitoes in each respective infection state. At time $t$, define
\begin{itemize}
    \item $h_{i,j}(t)$, $i,j \in \mathbbm{Z}_{\geq 0}$ to be the \textit{number} of humans with a hypnozoite reservoir of size $i$ and MOB precisely $j$, and
    \item $m_i(t)$, $i \in \{0, 1 \}$ to be the \textit{number} of uninfected and infected mosquitoes respectively.
\end{itemize}
We can thus denote the state of the Markov chain at time $t$ to be $( \mathbf{h}(t), \mathbf{m}(t))$ where
\begin{align*}
    \mathbf{h}(t) &= (h_{0,0}(t), h_{0,1}(t), h_{1,0}(t), h_{0, 2}(t), h_{1, 1,}(t), h_{2,0}(t), \dots)\\
    \mathbf{m}(t) &= (m_0(t), m_1(t))
\end{align*}
The state space is then
\begin{align*}
    \chi = \big\{ (\mathbf{h}, \mathbf{m}) \in [0, P_H]^{\mathbbm{N}} \times [0, P_M]^2: |\mathbf{h}|_1 = P_H, |\mathbf{m}|_1=P_M \big\},
\end{align*}
where the $\big( \frac{1}{2}(i+j+1)(i+j) + i +1\big)^{\text{th}}$ term of $\mathbf{h}(t)$ corresponds to $h_{i,j}(t)$. For notational convenience, we denote by $\mathbf{e_{i,j}}$ the $\big( \frac{1}{2}(i+j+1)(i+j) + i +1\big)^{\text{th}}$ (unit) coordinate vector in $\mathbbm{R}^\mathbf{N}$.\\

To obtain a density-dependent Markov population process, we proceed with the state description $(\mathbf{h}(t), \mathbf{m}(t))$. The transition rates are
\begin{align}
    &q_{ ( \mathbf{h}, \mathbf{m}), ( \mathbf{h} -  \mathbf{e_{i, j}} + \mathbf{e_{i, j-1}}, \mathbf{m})} = \gamma j h_{i, j}, \, \, i \geq 0, j \geq 1 \label{vivax_transition_rate_human_recov}\\
    &q_{ ( \mathbf{h}, \mathbf{m}), ( \mathbf{h} -  \mathbf{e_{i, j}} + \mathbf{e_{i-1, j}}, \mathbf{m})} = \mu i h_{i, j}, \, \, i \geq 1, j \geq 0\\
    &q_{ ( \mathbf{h}, \mathbf{m}), ( \mathbf{h} -  \mathbf{e_{i, j}} + \mathbf{e_{i-1, j+1}}, \mathbf{m})} = \alpha i h_{i, j}, \, \, i \geq 1, j \geq 0\\
    &q_{ ( \mathbf{h}, \mathbf{m}), ( \mathbf{h} -  \mathbf{e_{i, j}} + \mathbf{e_{i + k, j+1}}, \mathbf{m})} = \frac{\beta p \nu^k}{(\nu + 1)^{k+1}}  \frac{m_1}{P_M} h_{i, j}, \, \, i \geq 0, j \geq 0  \label{vivax_transition_rate_mos_human}\\
    &q_{( \mathbf{h}, \mathbf{m}), ( \mathbf{h}, \mathbf{m} - \mathbf{e_1} + \mathbf{e_{0}})} = g m_1\\
    &q_{( \mathbf{h}, \mathbf{m}), ( \mathbf{h}, \mathbf{m} - \mathbf{e_0} + \mathbf{e_{1}})}  = \beta q \Big( \sum^\infty_{i=0} \sum^\infty_{j=1} \frac{h_{i,j}}{P_H} \Big) m_0  \label{vivax_transition_rate_human_mos}
\end{align}
which can be understood as follows:
\begin{itemize}
    \item $q_{ ( \mathbf{h}, \mathbf{m}), ( \mathbf{h} -  \mathbf{e_{i, j}} + \mathbf{e_{i, j-1}}, \mathbf{m})}$: a human with hypnozoite reservoir size $i$ and MOB $j$ clears a single brood.
    \item $q_{ ( \mathbf{h}, \mathbf{m}), ( \mathbf{h} -  \mathbf{e_{i, j}} + \mathbf{e_{i-1, j}}, \mathbf{m})}$: a single hypnozoite dies in a human with hypnozoite reservoir size $i$ and MOB $j$.
    \item $q_{ ( \mathbf{h}, \mathbf{m}), ( \mathbf{h} -  \mathbf{e_{i, j}} + \mathbf{e_{i-1, j+1}}, \mathbf{m})}$: a single hypnozoite activates, thereby giving rise to a relapse (that is, a blood-stage infection with one additional brood), in a human with hypnozoite reservoir size $i$ and MOB $j$.
    \item $q_{ ( \mathbf{h}, \mathbf{m}), ( \mathbf{h} -  \mathbf{e_{i, j}} + \mathbf{e_{i + k, j+1}}, \mathbf{m})}$: a human with hypnozoite reservoir size $i$ and MOB $j$ gains an additional batch of $k$ hypnozoites, in addition to a primary infection (that is, a blood-stage infection with one additional brood) through the bite of an infected mosquito.
    \item $q_{( \mathbf{h}, \mathbf{m}), ( \mathbf{h}, \mathbf{m} - \mathbf{e_1} + \mathbf{e_{0}})}$: an infected mosquito dies to give rise to uninfected progeny.
    \item $q_{( \mathbf{h}, \mathbf{m}), ( \mathbf{h}, \mathbf{m} - \mathbf{e_0} + \mathbf{e_{1}})}$: an uninfected mosquito is infected by taking a bloodmeal from an infected human.
\end{itemize}

We set the ratio of the human and mosquito population sizes to be $v := \frac{P_M}{P_H}$, allowing us to write the transition rates given by Equations (\ref{vivax_transition_rate_human_recov}) to (\ref{vivax_transition_rate_human_mos}) in the form
\begin{align*}
    q_{( \mathbf{h}, \mathbf{m} ), ( \mathbf{h} + J_h, \mathbf{m} + J_m) } = P_M \cdot \omega_{(J_h, J_m)} \Big( \frac{\mathbf{h}}{P_M}, \frac{\mathbf{m}}{P_M} \Big) 
\end{align*}
where $\omega_{J_h, J_m}: \mathcal{R} \to \mathbbm{R}$ are given by
\begin{align}
    &\omega_{ (\mathbf{e_{i, j-1}} - \mathbf{e_{i, j},0)} }(\mathbf{H}, \mathbf{M}) = \gamma j H_{i, j}, \, \, i \geq 0, j \geq 1 \label{vivax_alpha_human_recov}\\
    &\omega_{ (\mathbf{e_{i-1, j}} -  \mathbf{e_{i, j}},0) } (\mathbf{H}, \mathbf{M}) = \mu i H_{i, j}, \, \, i \geq 1, j \geq 0\\
    &\omega_{( \mathbf{e_{i-1, j+1} - \mathbf{e_{i, j}} }, 0)}(\mathbf{H}, \mathbf{M}) = \alpha i H_{i, j}, \, \, i \geq 1, j \geq 0\\
    &\omega_{ ( \mathbf{e_{i + k, j+1}} -  \mathbf{e_{i, j}}, 0)}(\mathbf{H}, \mathbf{M}) = \frac{\beta p \nu^k}{(\nu + 1)^{k+1}}  M_1 H_{i, j}, \, \, i \geq 0, j \geq 0 \\
    &\omega_{( 0, \mathbf{e_{0}}- \mathbf{e_1} )}(\mathbf{H}, \mathbf{M}) = g M_1\\
    &\omega_{( 0, \mathbf{e_{1}}- \mathbf{e_0} )}(\mathbf{H}, \mathbf{M})  = \beta q v \Big( \sum^\infty_{i=0} \sum^\infty_{j=1} H_{i,j} \Big) M_0. \label{vivax_alpha_human_mos}
\end{align}
As such, we have defined a density-dependent Markov process with size parameter $P_M$.

\subsection{Steady state behaviour} \label{sec::vivax_markov_extinction}

The disease-free state is necessarily absorbing. However, since we have a density-dependent Markov population process with countably many types, absorption in the disease-free state is not guaranteed {\it a priori} since the mean MOB and/or hypnozoite reservoir size in the population can drift to infinity \parencite{luchsinger2001stochastic}.\\

In Theorem \ref{proof::vivax_markov_extinction} below, we show that the process does, in fact, reach the disease free state with probability one. To do this, we adopt the state description 
$(\mathbf{u}(t), \mathbf{v}(t))$ whereby each human is labelled with a hypnzoite/MOB state, and each mosquito is assigned a binary infection state.\\

If the arrivals of blood stage infections and batches of hypnozoites to the humans and infections to the mosquitoes were independent Poisson processes, then the the process $(\mathbf{u}(t),\mathbf{v}(t))$ could be regarded as a modelling $P_H$ independent batch arrival infinite server queueing processes with a structure as in Figure \ref{fig:vivax_queue} and $P_M$ independent on-off processes  whose $s^{th}$ component tells us whether the $s$th mosquito is infected or not.\\

However this is not quite the case. An `arrival' of an infection to a human depends on the number of infected mosquitoes and conversely, an `arrival' of an infection to a mosquito depends on the number of infected humans.\\

To overcome this, we couple the actual process $(\mathbf{u}(t),\mathbf{v}(t))$ with a second process $(\mathbf{u}'(t),\mathbf{v}'(t))$ in which humans (mosquitoes) become infected at constant rates independent of the number of mosquitoes (humans) that are infected. Specifically, in the process $(\mathbf{u}'(t),\mathbf{v}'(t))$, in modelling the rate at which humans become infected, we assume that all the mosquitoes are infected all of the time and, in modelling the rate at which mosquitoes become infected, we assume that all the humans are infected all of the time.\\

This amount of infection in the process $(\mathbf{u}'(t),\mathbf{v}'(t))$ can be shown to dominate the amount of infection in $(\mathbf{u}(t),\mathbf{v}(t))$. Furthermore, we can recognise $(\mathbf{u}'(t),\mathbf{v}'(t))$ as an independent network of infinite server queues for which a stability criterion is known, and hence establish that this criterion must apply to $(\mathbf{u}(t),\mathbf{v}(t))$ as well.\\

The details are given below.\\

\begin{theorem} \label{proof::vivax_markov_extinction}
With probability one, disease is eventually eliminated, that is,
\begin{align*}
    \lim_{t \to \infty} (\mathbf{u}(t), \mathbf{v}(t)) = (\mathbf{0}, \mathbf{0}),
\end{align*}
and the time to disease elimination has finite expectation. 
\end{theorem}

\begin{proof}

{To show that absorption in the disease-free state $(\mathbf{u}(t), \mathbf{v}(t))=(\mathbf{0}, \mathbf{0}$) occurs  almost surely, with a finite expected time to absorption, we adopt a coupling argument.}\\

{Consider an ensemble of $P_H$ \textit{independent} networks of infinite server queues, as defined in Section \ref{sec::vivax_queue_network}, each with homogeneous arrival rate $\beta p P_M/P_H$. Define the random vector $\mathbf{u'}(t) \in (\mathbbm{N} \times \mathbbm{N})^{P_H}$ to be the vector whose $r$th component is a vector encoding the numbers of individuals in compartment $H$, and the sum of individuals in compartments $A$ and $P$ of network or `human' $r$ at time $t$. Further, consider $P_M$ \textit{independent} continuous-time Markov chains, each with state space $\{ 0, 1\}$ and transition rate matrix
\begin{align*}
   Q:= \begin{pmatrix} -\beta q & \beta q \\ g & -g \end{pmatrix}
\end{align*}
and let the vector $\mathbf{v'}(t) \in \{0, 1\}^{P_M}$ denote the state of each chain or `mosquito' $s=1, \dots, P_M$ at time $t$.}\\ 

{We generate a coupling of the processes $\{(\mathbf{u}(t), \mathbf{v}(t)): t \geq 0 \}$ and $\{(\mathbf{u}'(t), \mathbf{v}'(t)): t \geq 0 \}$ as follows. Consider $P_M \times P_H$ independent homogeneous Poisson processes $\{ B_{rs}(t): t \geq 0 \}$, each of rate $\beta (p+q)/P_H$. We take $B_{rs}(t)$ to govern the sequence of interaction times between human $r$ and mosquito $s$ under both processes $\{(\mathbf{u}(t), \mathbf{v}(t)): t \geq 0 \}$ and $\{(\mathbf{u'}(t), \mathbf{v'}(t)): t \geq 0 \}$. The consequences of each human/mosquito interaction, however, can vary:
\begin{itemize}
    \item With probability $p/(p+q)$, a point in the Poisson process $B_{rs}(t)$ models a potential transmission event from mosquito $s$ to human $r$. If mosquito $s$ is infected at time $t$ in the process $\{(\mathbf{u}(t), \mathbf{v}(t)): t \geq 0 \}$, that is, $v_s(t) = 1$, then there is a coincident mosquito-to-human transmission event to human $r$ across both processes $\{(\mathbf{u}(t), \mathbf{v}(t)): t \geq 0 \}$ and $\{(\mathbf{u'}(t), \mathbf{v'}(t)): t \geq 0 \}$, with an equal hypnozoite batch size and an identical time course for each inoculated hypnozoite/primary infection. If, in contrast, $v_s(t) = 0$, then there is an arrival of a geometrically-distributed hypnozoite batch and a single primary infection (with an associated time course) into human $r$ under the process $\{(\mathbf{u'}(t), \mathbf{v'}(t)): t \geq 0 \}$ only.
    \item Otherwise, the point in the Poisson process $B_{rs}(t)$ models a potential transmission event from human $r$ to mosquito $s$. If human $r$ is blood-stage infected at time $t$ in the process $\{(\mathbf{u}(t), \mathbf{v}(t)): t \geq 0 \}$, that is, $(u_r)_2(t) \geq 1$, then mosquito $s$ immediately enters the infected state, where it remains for a common exponentially-distributed period of mean length $1/g$ under both processes $\{(\mathbf{u}(t), \mathbf{v}(t)): t \geq 0 \}$ and $\{(\mathbf{u'}(t), \mathbf{v'}(t)): t \geq 0 \}$. If, in contrast, $(u_r)_2(t) = 0$, then mosquito $s$ enters the infected state in the process $\{(\mathbf{u'}(t), \mathbf{v'}(t)): t \geq 0 \}$ only.
\end{itemize}}

{Under this setting, we necessarily have
\begin{align*}
    (\mathbf{u}(0), \mathbf{v}(0)) \leq  (\mathbf{u}'(0), \mathbf{v}'(0)) \implies (\mathbf{u}(t), \mathbf{v}(t)) \leq  (\mathbf{u}'(t), \mathbf{v}'(t)) \text{ for all } t \geq 0.
\end{align*}
In particular,
\begin{align*}
    (\mathbf{u}'(t), \mathbf{v}'(t)) = (\mathbf{0}, \mathbf{0}) \implies  (\mathbf{u}(t), \mathbf{v}(t)) = (\mathbf{0}, \mathbf{0}),
\end{align*}
and as such, the hitting time
\begin{align*}
    T_{(\mathbf{u_0}, \mathbf{v_0})} :=  \inf \big\{ \tau \geq 0: (\mathbf{u}'(\tau), \mathbf{v}'(\tau)) = (\mathbf{0}, \mathbf{0}) \, | \, (\mathbf{u}'(0), \mathbf{v}'(0)) = (\mathbf{u_0}, \mathbf{v_0}) \big\}
\end{align*}
yields an upper bound for the time to absorption in the disease-free state under the process $\{(\mathbf{u}(t), \mathbf{v}(t)): t \geq 0 \}$ with initial condition $(\mathbf{u}(0), \mathbf{v}(0)) = (\mathbf{u_0}, \mathbf{v_0})$.}\\

{Using Corollary 4.1.1 of \textcite{mehra2023open}, since the expected network occupation time for each hypnozoite/primary infection is finite and the hypnozoite batch size has finite mean, each component $\{ u'_r(t): t \geq 0 \}$, $r=1, \dots, P_H$ is ergodic. Further, since each Markov chain $\{ v'_s(t): t \geq 0 \}$, $s=1, \dots, P_M$ is irreducible and possesses a finite state space, it is also ergodic.}\\

{Each component $\{ u'_r(t): t \geq 0 \}$ and $\{ v'_s(t): t \geq 0 \}$ is positive recurrent and thus possesses a stationary distribution. The components evolve independently, therefore the stationary distribution for the multidimensional product $\{ (\mathbf{u'}(t), \mathbf{v'}(t)): t \geq 0\}$ should be given by the product of stationary distributions of each individual component. Since $\{ (\mathbf{u'}(t), \mathbf{v'}(t)): t \geq 0\}$ is non-explosive, it follows that it must be positive recurrent. This establishes that each state $(\mathbf{u_0}, \mathbf{v_0}) \in \chi'$ is positive recurrent and the return time to the disease-free state has finite expectation $\EX[T_{(\mathbf{0}, \mathbf{0})}] < \infty$.}\\

{Any state $(\mathbf{u_0}, \mathbf{v_0}) \in \chi'$ can be reached from the disease-free state $(\mathbf{0}, \mathbf{0})$ with positive probability $p(\mathbf{u_0}, \mathbf{v_0}) > 0$ \textit{prior} to return to the disease-free state through a concerted series of mosquito-inoculation events, with appropriate constraints on the time course of each infection and mosquito lifetimes. Consequently,
\begin{align*}
    \EX \big[ T_{(\mathbf{u_0}, \mathbf{v_0})} \big] \leq \sum^\infty_{n=1} p(\mathbf{u_0}, \mathbf{v_0}) \big( 1- p(\mathbf{u_0}, \mathbf{v_0}) \big)^{n-1} \cdot n \EX[T_{(\mathbf{0}, \mathbf{0})}] = \frac{\EX[T_{(\mathbf{0}, \mathbf{0})}]}{p(\mathbf{u_0}, \mathbf{v_0})} < \infty,
\end{align*}
that is, the expected hitting time for the disease-free state under the process $\{ (\mathbf{u'}(t), \mathbf{v'}(t)): t \geq 0\}$ has finite mean, irrespective of the initial condition $(\mathbf{u_0}, \mathbf{v_0}) \in \chi'$.}\\

{Since the time to absorption in the disease free state under the process $\{ (\mathbf{u}(t), \mathbf{v}(t)): t \geq 0\}$ with intial condition $(\mathbf{u_0}, \mathbf{v_0}) \in \chi'$ is bounded above by $T_{(\mathbf{u_0}, \mathbf{v_0})}$, it immediately follows that
\begin{align*}
    \lim_{t \to \infty} (\mathbf{u}'(t), \mathbf{v}'(t)) = (\mathbf{0}, \mathbf{0}) \text{ a.s } \implies \lim_{t \to \infty} (\mathbf{u}(t), \mathbf{v}(t)) = (\mathbf{0}, \mathbf{0}) \text{ a.s.},
\end{align*}
and that the time to disease elimination, moreover, has finite expectation.}\end{proof}

\subsection{A functional law of large numbers} \label{sec::vivax_lln_result}

We now use the work of \textcite{barbour2012law} to obtain a FLLN for the density-dependent Markov population process. In Theorem \ref{theorem::flln} below, we show that the sample paths of the continuous time Markov chain in converge to the solution of
\begin{align}
    \frac{d (\mathbf{H}, \mathbf{M})}{dt} = \sum_{J \in \mathcal{J}} \omega_J(\mathbf{H}, \mathbf{M}) \label{vivax_semilinear}
\end{align}
in the limit $P_M \to \infty$. Here, we define the $\eta$-norm
\begin{align*}
    || ( \mathbf{H},  \mathbf{M} ) ||_{\eta} := |M_0| + |M_1| + \sum^\infty_{i=0} \sum^\infty_{j=0} (i+j+1) |H_{i,j}|.
\end{align*}
and the set
\begin{align*}
    \mathcal{R}_\eta := \{ (\mathbf{H}, \mathbf{M}) \in 
    (\mathbbm{R}_{+})^{\mathbbm{
    N}} \times (\mathbbm{R}_{+})^2:  || ( \mathbf{H},  \mathbf{M} ) ||_{\eta} < \infty \}.
\end{align*}

To show that the the semilinear problem (\ref{vivax_semilinear}) with initial condition $(\mathbf{H}(0), \mathbf{M}(0)) \in \mathcal{R}_\eta$ has a unique mild solution in the interval $[0,\infty)$, we use the results of \textcite{barbour2012law} (which draw on Theorem 1.4, Chapter 6 of \textcite{pazy2012semigroups}), bounding the $\eta$-norm of $( \mathbf{H},  \mathbf{M} )$ using the expected occupancy of nodes $H$, $A$ and $P$ of the network introduced in Section \ref{sec::vivax_queue_network} (Appendix \ref{sec::weak_sol_exists}). To show the convergence of sample paths to the semilinear problem (\ref{vivax_semilinear}), we use Theorem 4.7 of \textcite{barbour2012law} after verifying a series of semigroup and transition rate assumptions (Appendices \ref{appendix::barbour_2.1} and \ref{appendix::barbour_4.2}).\\

\begin{theorem} \label{theorem::flln}
The semilinear problem (\ref{vivax_semilinear}) with initial condition $(\mathbf{H}(0), \mathbf{M}(0)) \in \mathcal{R}_\eta$ has a unique mild solution in the interval $[0,\infty)$. For each $T>0$, $\exists K_T^{(1)}, K_T^{(2)}, K_T^{(3)}$ such that if
\begin{align*}
    || P_M^{-1} ( \mathbf{h}(0), \mathbf{m}(0) ) - ( \mathbf{H}(0), \mathbf{M}(0) ) ||_{\eta} \leq K_T^{(1)} \sqrt{\frac{\log P_M}{P_M}} 
\end{align*}
for $P_M$ large enough, then
\begin{align*}
    P \Big[ \sup_{0 \leq t \leq T}  || P_M^{-1} (\mathbf{h}(t), \mathbf{m}(t)) -  (\mathbf{H}(t), \mathbf{M}(t)) ||_{\eta} > K_T^{(2)} \sqrt{\frac{\log P_M}{P_M}} \Big] \leq K_T^{(3)} \frac{\log{P_M}}{P_M}.
\end{align*}
\end{theorem}

\section{A deterministic population-level model} \label{sec::vivax_hybrid}

We devote the present section to the analysis of the semilinear system given by Equation (\ref{vivax_semilinear}), which arises as the FLLN limit \parencite{barbour2012law} for the density-dependent Markov population process detailed in Section \ref{sec::vivax_markov}.

\subsection{A countable system of ODEs} \label{sec::lln_ode}

 It is instructive to write out the components of the vector equation (\ref{vivax_semilinear}), rescaled to consider the \textit{proportion} of humans in each hypnozoite/MOB state. At time $t$, define
\begin{itemize}
    \item $H_{i,j}(t)$ to be the proportion of humans with hypnozoite reservoir size $i$ and MOB $j$
    \item $I_M(t)$ to be the proportion of infected mosquitoes.
\end{itemize}
Then we obtain the the infinite-dimensional system of ODEs
\begin{align}
    \frac{d H_{i, j}}{dt} =& - a I_M H_{i, j}  + a I_M \sum^{i}_{k=0} H_{k, j-1} \frac{1}{\nu + 1} \Big( \frac{\nu}{\nu + 1} \Big)^{i-k} \notag \\
    & - i (\mu + \alpha) H_{i, j} - j \gamma H_{i, j}  + (i+1) \mu H_{i+1, j} + (i+1) \alpha H_{i+1, j-1} + (j+1) \gamma H_{i, j+1} \label{vivax_hybrid_dhk}\\
    \frac{d I_M}{dt} =&  - g I_M + b \Big( \sum^\infty_{i=0} \sum^\infty_{j=1} H_{i,j} \Big) (1 - I_M) \label{vivax_hybrid_dim}
\end{align}
where, for notational convenience, we have set
\begin{align*}
    a = \frac{P_M}{P_H} \beta p \qquad 
    b = \beta q.
\end{align*}

The case $\nu=0$ yields the classical model of superinfection for \textit{P. falciparum}, as formulated by \textcite{bailey1957} (Appendix \ref{sec::bailey}).\\

Equation (\ref{vivax_hybrid_dhk}) constitutes a deterministic compartmental model for a vector that has total `mass' one. This invites us to draw a parallel with the Kolmogorov forward differential equations for a continuous-time Markov chain model with state space $\mathbbm{Z}^2$, where $H_{i,j}(t)$ has the interpretation that the Markov chain is in state $(i, j)$ at time $t$. Under the assumption that $aI_M(t)$ is the rate function of an inhomogeneous Poisson process, this model is precisely the open network of infinite server queues governing \textit{within-host} superinfection and hypnozoite dynamics that we discussed in Section \ref{sec::vivax_queue_network}, and Equation (\ref{vivax_hybrid_dhk}) is indeed the forward Kolmogorov differential equation for this Markov chain.\\

Rather than solving the system of ODEs (\ref{vivax_hybrid_dhk}) directly for some time dependent function $I_M(t)$, we can draw on a physical understanding of the queueing network to derive a generating function for the time dependent state probabilites $H_{i,j}(t)$ as a function of $I_M(\tau)$, $\tau \in [0, t)$. In particular, from Equation (\ref{vivax_multi_pgf}), given the absence of liver- and blood-stage infection in the human population at time zero, that is,
\begin{align*}
    H_{0,0}(0) = 1, \, H_{i,j}(0) = 0 \text{ for all } i + j >0
\end{align*}
it follows that
\begin{align}
    H(x, y, t) &:= \sum^{\infty}_{i=0} \sum^{\infty}_{j=0} H_{i, j}(t) x^i y^j \notag \\
    &= \exp \Big\{ -  a \int^{t}_0 I_M(\tau) \Big( 1 - \frac{\ 1 + (y-1) e^{-\gamma (t-\tau) }  }{1 + \nu (1 - x) p_H(t-\tau) + \nu (1 - y) p_A(t-\tau) } \Big) d \tau \Big\}, \label{vivax_H_queue}
\end{align}
where Equation (\ref{vivax_H_queue}) is guaranteed to converge in the domain $(x, y) \in [0, 1]^2$ for any fixed $t$.\\

We can exploit the functional dependence (\ref{vivax_H_queue}) between the generating function for the state probabilities $H_{i,j}(t)$ of the human population at time $t$, and the proportion of infected mosquitoes $I_M(\tau)$ for $\tau \in [0, t)$ to aid analysis of the infinite-compartment model given by Equations (\ref{vivax_hybrid_dhk}) and (\ref{vivax_hybrid_dim}).

\subsection{A reduced integrodifferential equation} \label{sec::lln_ide}
Upon examination of the deterministic model structure, we note that the coupling between human and mosquito dynamics is completely determined by the time evolution of
\begin{itemize}
    \item the prevalence of blood-stage infection in the human population (that is, humans with MOB at least one), which can be written
    \begin{align*}
        I_H(t) := \sum^\infty_{i=0} \sum^\infty_{j=1} H_{i, j}(t);
    \end{align*}
    \item the FORI $a I_M(t)$.
\end{itemize}

The utility of our observation in Section \ref{sec::lln_ode} emerges in the characterisation of the functional dependence between $I_H(t)$ and the FORI $a I_M(\tau)$ for $\tau \in [0, t)$. Specifically, given
\begin{align}
    H_{0,0}(0) = 1, \, H_{i,j}(0) = 0 \text{ for all } i + j >0, \label{ic::no_human_inf}
\end{align}
we can readily write
\begin{align}
    I_H(t) &= 1 - H(1, 0, t) = 1 - \exp \Big\{ -a \int^t_0 I_M(\tau) \Big( \frac{e^{-\gamma (t-\tau)} + \nu p_A(t-\tau)}{1 + \nu p_A(t-\tau)} \Big) d \tau \Big\}\label{eq:prev}
\end{align}
using the generating function (\ref{vivax_H_queue}).\\

Substituting the integral equation (\ref{eq:prev}) into the ODE (\ref{vivax_hybrid_dim}) yields the IDE
\begin{align}
    \frac{dI_M}{dt} &= -g I_M + b (1 - I_M) \bigg( 1 - \exp \Big\{ -a \int^t_0 I_M(\tau) \Big( \frac{e^{-\gamma (t-\tau)} + \nu p_A(t-\tau)}{1 + \nu p_A(t-\tau)} \Big) d \tau \Big\} \bigg). \label{ide_im}
\end{align}
When the human population is initially blood- and liver-stage infection-naive (that is, initial condition (\ref{ic::no_human_inf}) holds), the time evolution of the FORI $aI_M(t)$ is equivalent under the IDE (\ref{ide_im}), and the infinite-dimensional system of ODEs given by Equations (\ref{vivax_hybrid_dhk}) and (\ref{vivax_hybrid_dim}).\\

As a function of the FORI $a I_M(\tau)$, $\tau \in [0, t)$ which solves the IDE (\ref{ide_im}), we can use the generating function (\ref{vivax_H_queue}) to recover the population-level distribution of hypnozoite and superinfection states in the human population at time $t$. The recurrence relations provided in \textcite{mehra2023open} can be used to obtain the proportion of humans $H_{i,j}(t)$ in each hypnozoite/MOB state. The time-evolution of various quantities of biological interest can be recovered using the formulae provided in \textcite{mehra2022hypnozoite}.

\subsection{The steady state distribution} \label{sec::lln_ss}

Steady state analysis for the countably infinite system of ODEs given by Equations (\ref{vivax_hybrid_dhk}) and (\ref{vivax_hybrid_dim}) is not straightforward. Instead, we characterise the existence and asymptomatic stability of equilibria for the IDE (\ref{ide_im}).\\

Upon inspection, we find that an equilibrium solution $I_M^*$ of the IDE (\ref{ide_im}) must satisfy the integral equation
\begin{align}
    \frac{g I^*_M}{b(1-I^*_M )} = 1 - \exp \Big\{ -a I^*_M  \int^\infty_0 \Big( \frac{e^{-\gamma \tau} + \nu p_A(\tau)}{1 + \nu p_A(\tau)} \Big) d \tau \Big\}. \label{ide_im_ss}
\end{align}
The disease-free equilibrium $I_M^*=0$ clearly satisfies Equation (\ref{ide_im_ss}). On the domain $I_M^* \in [0, 1]$, we note that the LHS of (\ref{ide_im_ss}) is a monotonically increasing, convex function of $I_M^*$ that approaches infinity in the limit $I_M^* \to 1$ while the RHS is monotonically increasing, concave. The existence of a second solution to Equation (\ref{ide_im_ss}) therefore depends on the limiting slope as $I_M \to 0$. It exists if and only if the derivative of the LHS is less than the derivative of the RHS at $I_M^*=0$ (see Figure \ref{fig::gradient_schematic}).\\

\begin{figure}[h]
    \centering
    \includegraphics[width=0.65\textwidth]{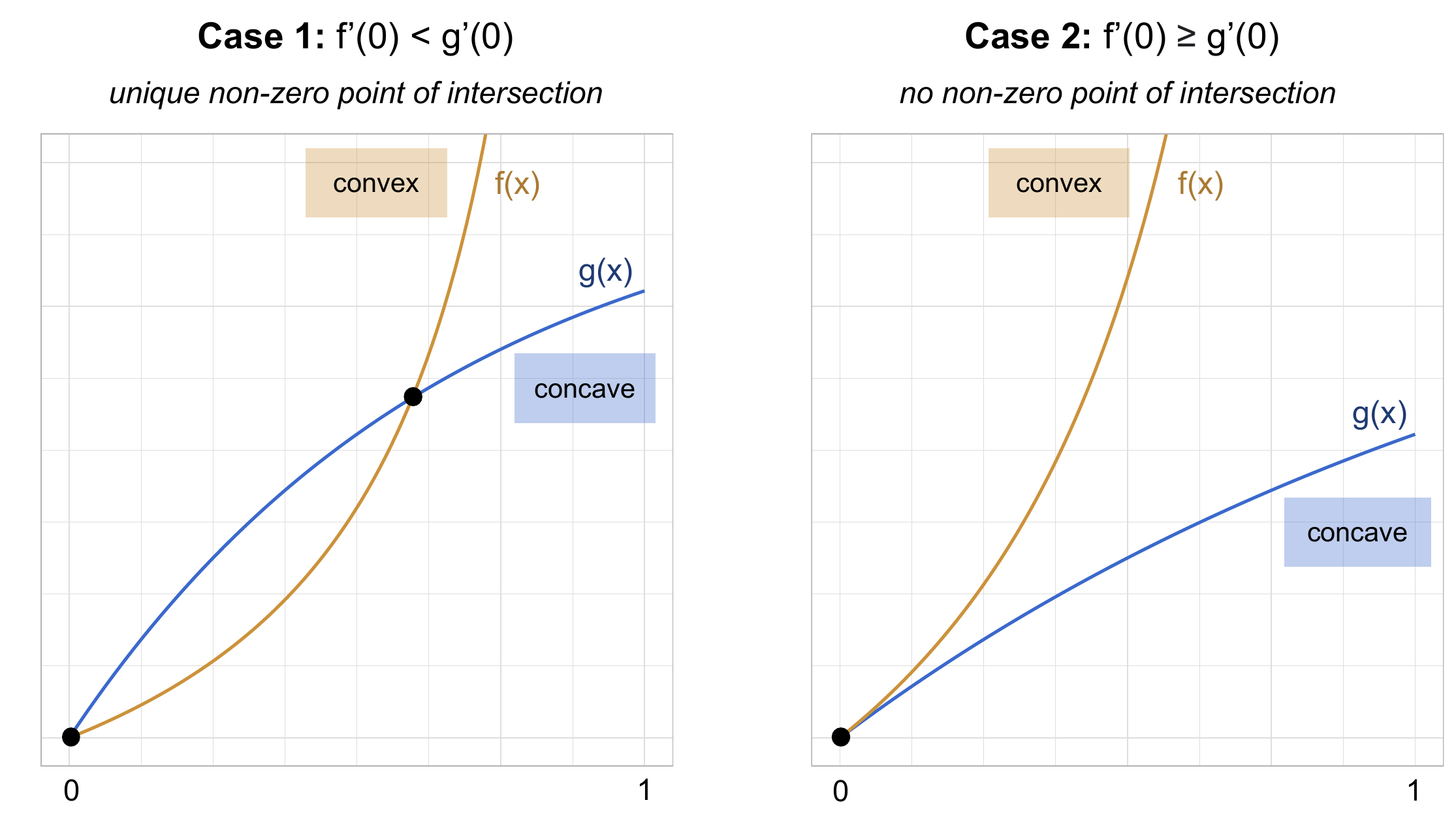}
    \caption{Schematic of the geometric argument to establish the existence of a non-zero solution to Equation (\ref{ide_im_ss}).}
    \label{fig::gradient_schematic}
\end{figure}

This allows us to identify a bifurcation parameter
\begin{align}
    R_0 := \sqrt{\frac{ab}{g} \Big( \int^\infty_0 \frac{e^{-\gamma \tau} + \nu p_A (\tau)}{1 + \nu p_A(\tau)} d \tau \Big)}, \label{eq:R0}
\end{align}
with $R_0>1$ a necessary and sufficient condition for the existence of an endemic equilibrium $I_M^*>0$.\\

We can readily interpret the integral expression in Equation (\ref{eq:R0}). The probability of a current blood-stage infection time $\tau$ after an infective bite is given by the integrand
\begin{align*}
    z(\tau) := 1 - \underbrace{ \vphantom{\frac{1}{\nu + p_A(\tau)}} (1 - e^{-\gamma \tau})}_{\substack{\text{no ongoing primary} \\ \text{infection}}} \cdot \underbrace{\frac{1}{1 + \nu p_A(\tau)}}_{\substack{\text{no ongoing relapses} \\ \text{(geometric batch)}}}.
\end{align*}
The integral
\begin{align*}
    \int^\infty_0 z(\tau) d \tau = \int^\infty_0 \frac{e^{-\gamma \tau} + \nu p_A (\tau)}{1 + \nu p_A(\tau)} d \tau
\end{align*}
therefore yields the expected cumulative duration of blood-stage infection attributable to a single infective bite. The functional form of the bifurcation parameter $R_0$ is analogous to that of a reproduction number \parencite{diekmann1990definition}.\\

Using the stability criterion of \textcite{brauer1978asymptotic}, whereby the IDE (\ref{ide_im}) is linearised about each equilibrium solution, we can also establish sufficient conditions for equilibrium solutions $I_M^*$ to be uniformly asymptotically stable. The threshold behaviour of the IDE (\ref{ide_im}) is summarised in Theorem \ref{theorem::ide_ss} below, with a proof provided in Appendix \ref{appendix:theorem_ide_ss}.

\vspace{3mm}

\begin{theorem}{(Threshold behaviour for the IDE (\ref{ide_im}))}
\begin{enumerate}[mode=unboxed]
    \item If $R_0 < 1$, then the disease-free equilibrium $I_M^* = 0$ is uniformly asymptotically stable and no endemic equilibrium solution exists.
    \item If $R_0 > 1$, there exists precisely one endemic equilibrium $I_M^* >0 $, given by the non-trivial solution to the equation
    \begin{align}
        1 - \frac{g I_M^*}{b(1-I_M^*)} - e^{-\frac{g}{b} R_0^2 I_M^*} = 0 \label{eq:ss_im*}
    \end{align}
     This endemic equilibrium $I^*_M>0$ is uniformly asymptotically stable if
    \begin{align*}
        I_M^* > \frac{1 + \frac{b}{b+g} - \sqrt{\big( 1 - \frac{b}{b+g} \big)^2 + 4 \frac{b}{b+g} \frac{1}{R_0^2} }}{2}.
    \end{align*}
\end{enumerate}
\end{theorem} \label{theorem::ide_ss}

\section{A general strategy for constructing tractable models of malarial superinfection} \label{sec::discussion}
Our conceptual approach for modelling superinfection and hypnozoite dynamics for vivax malaria has been three-fold.\\

To establish a functional relationship between the intensity of mosquito-to-human transmission, and the dynamics of superinfection and hypnozoite accrual on the within-host scale (Section \ref{sec::vivax_queue_network}), we construct an open network of infinite server queues that accounts for stochasticity in: 
\begin{itemize}
    \item the temporal sequence of mosquito-to-human transmission events, modelled with a non-homogeneous Poisson process;
    \item parasite inocula and consequently hypnozoite batch sizes, assumed to be geometrically-distributed for each bite \parencite{white2014modelling}; and
    \item the time course of each hypnozoite and primary infection, governed by independent stochastic processes. 
\end{itemize}
The independence structure of the queueing network enables us to characterise the occupancy distribution through relatively straightforward physical arguments \parencite{mehra2022hypnozoite}.\\

We model population level dynamics by constructing a Markov population process to addressing the dependence between the burden of blood-stage infection in the human population, and the intensity of mosquito-to-human transmission (Section \ref{sec::vivax_markov}). To characterise the steady state behaviour, we couple the Markov population process to an independent ensemble of queueing networks, with the same structure as the within host model and a known stability criterion (Theorem \ref{proof::vivax_markov_extinction}). Re-formulation of the state space to count the number of humans and mosquitoes in each permissible infection state, followed by a rescaling argument, yields a density-dependent Markov population process for which a FLLN limit, taking the form of a countably-infinite system of ODEs, is recovered using the work of \textcite{barbour2012law} (Theorem \ref{theorem::flln}).\\

We pay particular attention to the FLLN limit, recognising that it takes a form identical to what could be derived using a deterministic compartmental model of the type that are ubiquitous in mathematical epidemiology (Section \ref{sec::vivax_hybrid}). By recognising that those parts of the compartmental model that relate to the infection status of the human population are identical to the Kolmogorov forward differential equations for a network of infinite-server queues, we propose, to the best of our knowledge, a novel reduction to collapse the countable system of ODEs into a reduced IDE that is more amenable to analysis, and for which a threshold phenomenon can be characterised (Theorem \ref{theorem::ide_ss}).\\

A common vein in the analysis of population-level models through the course of this manuscript is the utility of a physical understanding of the underlying within-host model. We now seek to illustrate how the ideas presented in this manuscript can be adopted more generally to construct tractable models of malarial superinfection, with appropriate constraints on the model structure. 

\subsection{From a countable system of ODEs to a reduced integrodifferential equation}

We arrived at the countable system of ODEs given by Equations (\ref{vivax_hybrid_dhk}) and (\ref{vivax_hybrid_dim}) by recovering a FLLN for an appropriately-scaled density dependent Markov population process. Following a standard compartment modelling approach, however, we can view this system as a natural extension of the Ross-Macdonald framework to allow for superinfection \parencite{bailey1957} and hypnozoite accrual \parencite{white2014modelling}.\\

We can also derive Equations (\ref{vivax_hybrid_dhk}) and (\ref{vivax_hybrid_dim}) under the ``hybrid approximation'', as delineated by \textcite{nasell2013hybrid, henry2020hybrid}. Equation (\ref{vivax_hybrid_dhk}), which captures the time evolution of the within-host PMF, can also be interpreted as the governing equation for the \textit{expected} frequency distribution in the human population, as a function of the FORI. Similarly, Equation (\ref{vivax_hybrid_dim}) is precisely the Kolmogorov forward differential equation equation governing the time evolution of the probability that each mosquito in the population is infected --- which can likewise be interpreted as the expected proportion of infected mosquitoes --- as a function of the prevalence of blood-stage infection in the human population. The premise of the ``hybrid approximation'' is to allow for stochasticity within the human and mosquito populations respectively, but to couple host and vector dynamics through expected values \parencite{nasell2013hybrid, henry2020hybrid}; the hybrid approximation to the present model of hypnozoite accrual and superinfection yields precisely Equations (\ref{vivax_hybrid_dhk}) and (\ref{vivax_hybrid_dim}). Parallels between the FLLN and hybrid approximation have been noted previously in the literature: according to \textcite{hadeler1983nonlinear}, \textcite{rost1979method} proves that hybrid and mean-field approximations are equivalent\footnote{We have only been able to find the citation for \textcite{rost1979method}, but not the actual conference paper}, while \textcite{lewis1975model} establishes similar results through simulation.\\

Irrespective of the derivation of the infinite-compartment model it is the observation that Equation (\ref{vivax_hybrid_dhk}) constitutes the Kolmogorov forward differential equations for an open network of infinite server queues (Section \ref{sec::lln_ode}) that constitutes the crux of our analysis. Rather than solving the system of ODEs (\ref{vivax_hybrid_dhk}) directly, we use a physical understanding of the queueing structure to derive the PGF (\ref{vivax_H_queue}). From the PGF, we extract the prevalence of blood-stage infection as a function of the FORI --- which underpins the coupling between host and vector dynamics, and enables the derivation of a single IDE governing the FORI; in addition to quantities of biological interest, the time evolution of which can be recovered as a function of the FORI solving the aforementioned IDE.\\

Challenges posed by infinite-dimensional systems of ODEs for simulation and analysis have prompted the development of methods like the pseudoequilibrium approximation to yield finite-dimensional systems of ODEs in the presence of superinfection \parencite{dietz1974malaria, henry2020hybrid}. Our IDE reduction offers an alternative construction that yields an identical FORI to the original system of ODEs, under the assumption that the human population is initially uninfected.\\

Relative to the countable system of ODEs which served as our starting point, the reduced IDE is more amenable to steady state analysis. For compartment models with finitely many types (comprising finite-dimensional systems of ODEs), the next generation matrix method of \textcite{diekmann2010construction} is routinely used to characterise threshold phenomena: the basic reproduction number $R_0$ --- which quantifies the number of additional cases stemming from an index case --- is computed as the spectral radius of a ``next generation operator'', with the key implication that $R_0 < 1$ is a necessary and sufficient condition for the stability of the trivial (disease-free) equilibrium \parencite{heffernan2005perspectives}. While the notion of a reproduction number and spectral bound generalises to epidemic models with countably many types \parencite{thieme2009spectral}, computation is not straightforward. In the present manuscript, we apply the stability criterion of \textcite{brauer1978asymptotic} to perform an asymptotic stability analysis for the reduced IDE, but are unaware of readily-verifiable stability criteria applicable to the countable system of ODEs. 

\subsection{Analyticity at the within-host scale}

The strategy delineated above is underpinned by the fact that the within-host model is analytically-tractable. The critical assumption is that consequences of each mosquito-to-human transmission event are governed by independent stochastic processes, for which a time-dependent PGF can be derived. We can accommodate multiple `types' of incoming parasites/infections (here, we distinguish hypnozoites from immediately developing parasite forms, represented by a primary infection), and occupancy times for each parasite/lifecycle stage do not necessarily have to be Markovian (that is, we can allow for general service time distributions). Tractable extensions to the within-host model that permit the IDE reduction proposed here encompass the class of models analysed in \textcite{mehra2023open}, comprising open networks of infinite server queues with nonhomogeneous multivariate batch Poisson arrivals. They are applicable when the evolution of `individuals' (hypnozoites or blood-stage infections in this case) occurs independently. Importantly, this is not the case when there is competition or another form of interaction between parasite broods. Further, while service time distributions and arrival rates can vary deterministically with time, they should not be dependent on the within-host state, barring the consideration of certain forms of immunity.\\

Nonetheless, our work provides a theoretical basis to explore the epidemiology of malarial superinfection, with natural extensions to allow for greater biological realism. A generalisation of the present framework, allowing for long-latency hypnozoites \parencite{white2014modelling, mehra2020activation}; time-varying transmission parameters, and the acquisition of transmission-blocking and antidisease immunity for \textit{P. vivax}, is presented in \textcite{mehra2022hybrid}. There, we derive an infinite compartment model under a ``hybrid approximation'' \parencite{nasell2013hybrid, henry2020hybrid}, comprising an infinite-dimensional system of ODEs, from which we derive a reduced system of IDEs using the approach detailed in Section \ref{sec::vivax_hybrid}. Analyses of both steady state solutions and transient dynamics under this generalised deterministic model are tractable using the within-host distributions derived in \textcite{mehra2022hypnozoite}. 

\section*{Acknowledgements}
The authors thank Jennifer Flegg for valuable discussions around the incorporation of superinfection in population-level models. JMM acknowledges support from the Australian Research Council through ARC Discovery Project DP210101920. PGT acknowledges support from the Australian Research Council through Laureate Fellowship FL130100039.



\newpage

\appendix

\section{A functional law of large numbers}\label{sec::vivax_lln}

We now verify the conditions for the FLLN \textcite{barbour2012law} of the Markov population process governing hypnozoite and superinfection dynamics, as described in Section \ref{sec::vivax_markov}. The scaled transition rates $\omega_J$ are given by Equations (\ref{vivax_alpha_human_recov}) to (\ref{vivax_alpha_human_mos}). We define
\begin{align*}
    \mathcal{R} & :=  (\mathbbm{R}_{+})^{\mathbbm{
    N}} \times (\mathbbm{R}_{+})^2
\end{align*}
with each $(\mathbf{H}, \mathbf{M}) \in \mathcal{R}$ indexed in the form
\begin{align*}
    (\mathbf{H}, \mathbf{M}) = \Big( (H_{0,0}, H_{0,1}, H_{1,0}, H_{0,2}, H_{1,1}, H_{2,0}, \dots ), \, (M_0, M_1) \Big),
\end{align*}
in addition to the subset
\begin{align*}
    \mathcal{R}_0 &:= \{ (\mathbf{H}, \mathbf{M}) \in \mathcal{R}: H_{i,j} = 0 \text{ for all but finitely many } i, j \}.
\end{align*}

Let $\mathcal{J}$ be the set of jumps with non-zero transition rates. Jumps are of ``bounded influence'' in the sense of \parencite{barbour2012law}, in that at most two compartments are affected. Further,
\begin{align*}
    \sum_{J \in \mathcal{J}} \omega_J (\mathbf{H}, \mathbf{M}) < \infty \text{ for all } (\mathbf{H}, \mathbf{M}) \in \mathcal{R}_0.
\end{align*}
Therefore, Assumptions 1.2 and 1.3 of \textcite{barbour2012law} are satisfied, with the implication that Markov process with transition rates given by Equations (\ref{vivax_alpha_human_recov}) to (\ref{vivax_alpha_human_mos}) is a ``pure jump process, at least for some non-zero length of time'' \parencite{barbour2012law}.\\

In the notation of \textcite{barbour2012law}, we write the proposed limit as a semilinear equation
\begin{align*}
    \frac{d (\mathbf{H}, \mathbf{M})}{dt} = \sum_{J \in \mathcal{J}} \omega_J(\mathbf{H}, \mathbf{M}) =  A \cdot (\mathbf{H}, \mathbf{M}) + F \big( (\mathbf{H}, \mathbf{M}) \big)
\end{align*}
where $A$ is a constant (that is, not dependent on time) matrix, with non-zero terms
\begin{align*}
    A_{H_{i,j}, H_{i,j}} & = -i (\mu + \alpha) - j \gamma\\
    A_{H_{i-1,j}, H_{i,j}} & = i \mu \\
    A_{H_{i,j-1}, H_{i,j}} & = j \gamma\\
    A_{H_{i-1,j+1}, H_{i,j}} & = i \alpha\\
    A_{M_1, M_1} &= -g\\
    A_{M_0, M_1} &= g,
\end{align*}
while $F$ has components
\begin{align*}
    F^{H_{i,0}} & = -\beta p M_1 H_{i,0}\\
    F^{H_{i,j}} & = -\beta p M_1 H_{i,j} + \beta p M_1 \sum^{i}_{k=0} H_{k, j-1} \frac{1}{\nu + 1} \Big(\frac{\nu}{\nu +1} \Big)^{i-k}\\
    F^{M_0} &= -\beta q v \Big( \sum^\infty_{j=1} \sum^\infty_{i=0} H_{i, j} \Big) M_0\\
    F^{M_1} &= \beta q v \Big( \sum^\infty_{j=1} \sum^\infty_{i=0} H_{i, j} \Big) M_0.
\end{align*}

\subsection{Assumption 2.1 of \textcite{barbour2012law}} \label{appendix::barbour_2.1}

We recapitulate and verify Assumption 2.1 of \textcite{barbour2012law} below, with changes in notation/exposition in line with the current manuscript.\\

We begin by introducing the set
\begin{align*}
    \mathcal{X}_{+} &:= \Big\{ (\mathbf{H}, \mathbf{M}) \in (\mathbbm{Z}_{\geq 0})^{\mathbbm{
    N}} \times (\mathbbm{Z}_{\geq 0})^2: M_0 + M_1 + \sum^\infty_{i=0} \sum^\infty_{j=0} H_{i,j} < \infty \Big\}.
\end{align*}

By assigning a `size' to each human and mosquito infection state through a fixed vector $\eta = (\eta^{(h)}, \eta^{(m)}) \in \mathcal{R}$, such that
\begin{align}
    \eta^{(h)}_{i,j},  \eta^{(m)}_{k} \geq 1 \qquad \lim_{i+j \to \infty} \eta^{(h)}_{i,j} = \infty, \label{eq:eta_size}
\end{align}
we can construct an empirical moment analogue
\begin{align*}
    S_r \big( (\mathbf{H}, \mathbf{M}) \big) &= \big( \eta^{(m)}_0 \big)^r M_0 + \big( \eta^{(m)}_1 \big)^r M_1 + \sum^\infty_{i=0} \sum^\infty_{j=0} \big( \eta^{(h)}_{i,j} \big)^r H_{i,j}, \, (\mathbf{H}, \mathbf{M}) \in \mathcal{R}_0.
\end{align*}

To bound these empirical moments, \textcite{barbour2012law} introduce the quantites
\begin{align*}
    U_r \big( (\mathbf{H}, \mathbf{M}) \big) &= \sum_{J \in \mathcal{J}} \omega_{J}(\mathbf{H}, \mathbf{M}) \cdot \bigg\{ J^{(m)}_0 \big( \eta^{(m)}_0 \big)^r + J^{(m)}_1 \big( \eta^{(m)}_1 \big)^r + \sum^\infty_{i=0} \sum^\infty_{j=0} J^{(h)}_{i,j} \big( \eta^{(h)}_{i,j} \big)^r \bigg\}\\
    V_r \big( (\mathbf{H}, \mathbf{M}) \big) &= \sum_{J \in \mathcal{J}} \omega_{J}(\mathbf{H}, \mathbf{M}) \cdot \bigg\{ J^{(m)}_0 \big( \eta^{(m)}_0 \big)^r + J^{(m)}_1 \big( \eta^{(m)}_1 \big)^r + \sum^\infty_{i=0} \sum^\infty_{j=0} J^{(h)}_{i,j} \big( \eta^{(h)}_{i,j} \big)^r \bigg\}^2.
\end{align*}
In a similar vein, we set
\begin{align*}
    W_r \big( (\mathbf{H}, \mathbf{M}) \big) &= \sum_{J \in \mathcal{J}} \omega_{J}(\mathbf{H}, \mathbf{M}) \cdot \bigg| J^{(m)}_0 \big( \eta^{(m)}_0 \big)^r + J^{(m)}_1 \big( \eta^{(m)}_1 \big)^r + \sum^\infty_{i=0} \sum^\infty_{j=0} J^{(h)}_{i,j} \big( \eta^{(h)}_{i,j} \big)^r \bigg|
\end{align*}
Assumption 2.1 of \textcite{barbour2012law} then reads as follows.\\

\begin{assumption} \label{assumption_2.1}
There exists $\eta \in \mathcal{R}$ satisfying condition (\ref{eq:eta_size}) and $r_\text{max}^{(1)}, r_\text{max}^{(2)} \geq 1$ such that the following conditions hold:
\begin{enumerate}
    \item For all $(\mathbf{H}, \mathbf{M}) \in \mathcal{X}_+$
    \begin{align*}
         W_r \Big( |\mathbf{M}|_1^{-1} (\mathbf{H}, \mathbf{M})  \Big) < \infty, \qquad 0 \leq r \leq r^{(1)}_\text{max}.
    \end{align*}
    \item $\exists$ non-negative constants $k_{rl}$ such that, for all $(\mathbf{H}, \mathbf{M}) \in \mathcal{R}$ with $|\mathbf{H}|_1 = 1/v, \, |\mathbf{M}|_1 = 1$,
    \begin{align*}
        U_0 \big( ( \mathbf{H}, \mathbf{M}) \big) &\leq k_{01} S_0 \big( ( \mathbf{H}, \mathbf{M}) \big) + k_{04}\\
        U_1 \big( ( \mathbf{H}, \mathbf{M}) \big) &\leq k_{11} S_1 \big( ( \mathbf{H}, \mathbf{M}) \big) + k_{14}\\
        U_r \big( ( \mathbf{H}, \mathbf{M}) \big) &\leq \{ k_{r1} + k_{r2} S_0 \big( ( \mathbf{H}, \mathbf{M}) \big) \} \cdot S_r \big( ( \mathbf{H}, \mathbf{M}) \big) + k_{r4}, \, 2 \leq r \leq r^{(1)}_\text{max}.
    \end{align*}
    \item $\exists$ non-negative constants $k_{rl}$ and a function $p(r)$ with $1 \leq p(r) \leq r_\text{max}^{(1)}$ for $1 \leq r \leq r_\text{max}^{(2)}$, such that, for all $(\mathbf{H}, \mathbf{M}) \in \mathcal{R}$ with $|\mathbf{H}|_1 = 1/v, \, |\mathbf{M}|_1 = 1$
    \begin{align*}
        V_0 \big( ( \mathbf{H}, \mathbf{M}) \big) &\leq k_{03} S_0 \big( ( \mathbf{H}, \mathbf{M}) \big) + k_{05}\\
        V_r \big( ( \mathbf{H}, \mathbf{M}) \big) &\leq k_{r3} S_{p(r)} \big( ( \mathbf{H}, \mathbf{M}) \big) + k_{r5}.
    \end{align*}
\end{enumerate}
\end{assumption}

\begin{claim}
Assumption \ref{assumption_2.1} holds for $\eta = (\eta^{(h)}, \eta^{(m)}) \in \mathcal{R}$ with
\begin{align}
    \eta^{(h)}_{i,j} = i + j + 1, \, \eta^{(m)}_k = 1, \label{eq:eta_choice}
\end{align}
and arbitrarily large $r_\text{max}^{(1)}, r_\text{max}^{(2)}$.
\end{claim}
\begin{proof}
Given our choice of $\eta$, we can write
\begin{align*}
    S_r \big( ( \mathbf{H}, \mathbf{M} ) \big) = & M_0 + M_1 + \sum^\infty_{i=0} \sum^\infty_{j=0} (i+j+1)^r H_{i,j}\\
    U_r \big( (\mathbf{H}, \mathbf{M} ) \big) = & \sum^\infty_{i=0} \sum^\infty_{j=0} ( \gamma j + \mu i \big) H_{i,j} \big[ (i + j)^r - (i+j+1)^r \big] \\
    & + \beta p M_1 \sum^\infty_{i=0} \sum^\infty_{j=0} H_{i,j} \sum^\infty_{k=0} \frac{1}{\nu + 1} \Big( \frac{\nu}{\nu+1} \Big)^k \cdot \big[ (i + j + k +2)^r - (i + j + 1)^r \big] \\
    V_r \big( (\mathbf{H}, \mathbf{M} ) \big) =&  \sum^\infty_{i=0} \sum^\infty_{j=0} ( \gamma j + \mu i \big) H_{i,j} \big[ (i + j)^r - (i+j+1)^r \big]^2 \\
    & + \beta p M_1 \sum^\infty_{i=0} \sum^\infty_{j=0} H_{i,j} \sum^\infty_{k=0} \frac{1}{\nu + 1} \Big( \frac{\nu}{\nu+1} \Big)^k \cdot \big[ (i + j + k +2)^r - (i + j + 1)^r \big]^2\\
    W_r \big( (\mathbf{H}, \mathbf{M} ) \big) =&  \sum^\infty_{i=0} \sum^\infty_{j=0} ( \gamma j + \mu i \big) H_{i,j} \big[ (i + j + 1)^r - (i+j)^r \big] \\
    & + \beta p M_1 \sum^\infty_{i=0} \sum^\infty_{j=0} H_{i,j} \sum^\infty_{k=0} \frac{1}{\nu + 1} \Big( \frac{\nu}{\nu+1} \Big)^k \cdot \big[ (i + j + k +2)^r - (i + j + 1)^r \big]
\end{align*}
We begin by noting that
\begin{align*}
    U_0 \big( (\mathbf{H}, \mathbf{M} ) \big) &=  V_0 \big( (\mathbf{H}, \mathbf{M} ) \big) = 0.
\end{align*}

In the case $r=1$, we can simplify these expressions substantially to yield the bounds
\begin{align*}
    U_1 \big( (\mathbf{H}, \mathbf{M} ) \big) &= \sum^\infty_{i=0} \sum^\infty_{j=0} \Big[ \beta p (\nu + 1) M_1 - \gamma j - \mu i \Big] H_{i,j} \leq \frac{\beta p(\nu + 1)}{v}\\
    V_1 \big( (\mathbf{H}, \mathbf{M} ) \big) &= \sum^\infty_{i=0} \sum^\infty_{j=0} \Big[ \beta p (1 + 3\nu + 2 \nu^2) M_1 + \gamma j + \mu i \Big] H_{i,j}\\
    & \leq \frac{\beta p(1 + 3 \nu + 2 \nu^2)}{v} + (\gamma + \mu) S_0 \big( (\mathbf{H}, \mathbf{M}) \big)
\end{align*}
for all $(\mathbf{H}, \mathbf{M} ) \in \mathcal{R}$ with $|\mathbf{H}|_1 = 1/v, \, |\mathbf{M}|_1 = 1$.\\

For $r \geq 1$ and $(\mathbf{H}, \mathbf{M} ) \in \mathcal{R}$ with $|\mathbf{H}|_1 = 1/v, \, |\mathbf{M}|_1 = 1$, we obtain the bounds
\begin{align*}
    U_r \big( (\mathbf{H}, \mathbf{M}) \big) &\leq \beta p M_1  \cdot \sum^\infty_{i=0} \sum^\infty_{j=0} H_{i,j} (i + j +1)^{r} \cdot \Bigg\{ \sum^\infty_{k=0} \frac{1}{1+\nu} \Big(  \frac{1}{1+\nu} \Big)^k \Big( 1 + \frac{k+1}{i+j+1} \Big)^r \Bigg\} \\
    & \leq \beta p \EX \big[ \big( H_b + 2 \big)^r \big] \cdot S_r \big( (\mathbf{H}, \mathbf{M}) \big)\\
    V_r \big( (\mathbf{H}, \mathbf{M} ) \big) & \leq  \sum^\infty_{i=0} \sum^\infty_{j=0}  H_{i,j} (i + j +1)^{2r} \Bigg[ \gamma j + \mu i + \beta p M_1 \cdot  \Bigg\{ \sum^\infty_{k=0} \frac{1}{1+\nu} \Big(  \frac{1}{1+\nu} \Big)^k \Big( 1 + \frac{k+1}{i+j+1} \Big)^{2r} \Bigg\} \Bigg]\\
    & \leq \Big\{ (\gamma + \mu) + \beta p  \EX \big[ \big( H_b + 2 \big)^{2r} \big] \Big\} \cdot  S_{2 r+1} \big( (\mathbf{H}, \mathbf{M} ) \big)
\end{align*}
where $H_b$ is a geometrically-distributed random variable with mean $\nu$ and state space $\mathbbm{Z}_{\geq 0}$, noting that $\EX[H_b^p] < \infty$ for all $p \geq 0$.\\

Likewise, for a given value $(\mathbf{H}, \mathbf{M}) \in \mathcal{X}_+$, 
\begin{align*}
    W_r \Big( |\mathbf{M}|_1^{-1} (\mathbf{H}, \mathbf{M})  \Big) & \leq \Big\{ (\gamma + \mu) +  \beta p \EX \big[ \big( H_b + 2 \big)^r \big] \Big\} \cdot \sum^\infty_{i=0} \sum^\infty_{j=0} \frac{H_{i,j}}{|\mathbf{M}|_1} (i+j+1)^{r+1} < \infty
\end{align*}
for all $r > 0$ since the RHS constitutes a finite sum. We thus see that Assumption \ref{assumption_2.1} is satisfied for arbitrarily large $r_\text{max}^{(1)}, r_\text{max}^{(2)} \geq 1$.
\end{proof}

\subsection{Existence of a unique weak solution to the semilinear equation (\ref{vivax_semilinear})} \label{sec::weak_sol_exists}

We now establish the existence and uniqueness of weak solutions to the semilinear equation (\ref{vivax_semilinear}) (which can be written explicitly as the system of ODEs given by Equations (\ref{vivax_hybrid_dhk}) and (\ref{vivax_hybrid_dim}) after re-scaling $H_{i,j}$ by a factor of $v$).\\

We first verify Assumption 3.1 of \textcite{barbour2012law}, by noting that:
\begin{align*}
    &A_{H_{i, j}, H_{k, \ell}} \geq 0 \text{ for all } (i, j) \neq (k, \ell); \quad A_{M_{k}, M_{\ell}} \geq 0 \text{ for all } k \neq \ell; \quad A_{H_{i,j}, M_k}=A_{M_k, H_{i,j}} = 0\\
    & \sum_{(k, \ell) \neq (i, j)} A_{H_{i,j}, H_{k, \ell}} = (\alpha + \mu)i + \gamma j < \infty \text{ for all } i,j \geq 0.
\end{align*}

We are then required to fix $\mu \in \mathcal{R}$ with $\mu^{(h)}_{i,j}, \mu^{(m)}_k \geq 1$, with the constraint that:
\begin{enumerate}
    \item there exists $r \leq r_\text{max}^{(2)}$ such that $\sup_{i, j \geq 0} \{ \mu^{(h)}_{i,j}/(\eta^{(h)}_{i,j})^r \} < \infty$ (Assumption 4.2.1 of \parencite{barbour2012law})
    \item $A^T \mu \leq w \mu$ for some $w \geq 0$ (Assumption 3.2 of \parencite{barbour2012law})
\end{enumerate}

Since Assumption \ref{assumption_2.1} is satisfied for arbitrarily large $r^{(2)}_\text{max}$, it suffices to set $\mu=\eta$, in which case we can choose $w = \max \{ \alpha + \mu, \gamma, g\}$. We thus define the $\eta$-norm
\begin{align*}
    || ( \mathbf{H},  \mathbf{M} ) ||_{\eta} := |M_0| + |M_1| + \sum^\infty_{i=0} \sum^\infty_{j=0} (i+j+1) |H_{i,j}|
\end{align*}
on the set
\begin{align*}
    \mathcal{R}_\eta := \{ (\mathbf{H}, \mathbf{M}) \in \mathcal{R}: || ( \mathbf{H},  \mathbf{M} ) ||_{\eta} < \infty \}.
\end{align*}

The conditions of Theorem 3.1 of \textcite{barbour2012law} are thus satisfied. The implications of this theorem are as follows. As in \textcite{barbour2012law}, denote by $P(\cdot)$ ``the semigroup of Markov transition matrices corresponding to the minimal process associated with $Q$'', where
\begin{align*}
    &Q_{H_{i,j}, H_{\ell,k}} = A^T_{H_{i,j}, H_{k, \ell}} \Big( \frac{k + \ell + 1}{i + j + 1} \Big) - w \delta_{\{ (i, j), (k, \ell) \}}\\
    &Q_{H_{i,j}, M_{k}} = Q_{M_k, H_{i,j}} = 0\\
    &Q_{M_{k}, M_{\ell}} = A^T_{M_k, M_\ell} - w \delta_{\{k, \ell\}}
\end{align*}
and set
\begin{align*}
    &R^T_{H_{i,j}, H_{\ell,k}} = e^{wt} P_{H_{i,j}, H_{k, \ell}}(t) \Big( \frac{k + j + 1}{k + \ell + 1} \Big)\\
    &R^T_{H_{i,j}, M_{k}} = e^{wt} P_{H_{i,j}, M_k}(t) \big( i + j + 1 \big)\\
    &R^T_{M_k, H_{i,j}} = e^{wt} P_{M_k, H_{i,j}}(t) \Big( \frac{1}{ i + j + 1} \Big)\\
    &R^T_{M_k, M_\ell} = e^{wt} P_{M_k, M_\ell}(t)
\end{align*}

Then $R$ is a strongly continuous semigroup on $\mathcal{R}_\eta$, with $A = R'(0)$ and $R'(t) = R(t)A$ for all $t \geq 0$.\\

It remains to verify Assumption 4.1 of \textcite{barbour2012law}, which mandates that $F$ is locally-Lipschitz in the $\eta$-norm. For $(\mathbf{H}^{(1)}, \mathbf{M}^{(1)}) \neq (\mathbf{H}^{(2)}, \mathbf{M}^{(2)})$ such that 
\begin{align*}
    \big| \big|  F \big( (\mathbf{H}^{(1)}, \mathbf{M}^{(1)}) \big) \big| \big|_\eta , \big| \big|  & F \big( (\mathbf{H}^{(2)}, \mathbf{M}^{(2)}) \big) \big| \big|_\eta \leq z,
\end{align*}
using the inequality
\begin{align*}
    |ax - by| \leq |a| |x-y| + |b-a| |y|
\end{align*}
we can show that
\begin{align*}
    \big| \big|  & F \big( (\mathbf{H}^{(1)}, \mathbf{M}^{(1)}) \big) - F \big( (\mathbf{H}^{(2)}, \mathbf{M}^{(2)}) \big) \big|\big|_\eta\\
     =& 2 \beta q v \Bigg| M^{(1)}_0 \Bigg( \sum^\infty_{j=1} \sum^\infty_{i=0} H^{(1)}_{i,j} \Bigg) - M^{(2)}_0 \Bigg( \sum^\infty_{j=1} \sum^\infty_{i=0} H^{(1)}_{i,j} \Bigg) \Bigg| + \beta p \sum^\infty_{i=0} (i + 1) \cdot \Big| H^{(1)}_{i,0} M^{(1)}_1 - H^{(2)}_{i,0} M^{(2)}_1 \Big|\\
    & + \beta p \sum^\infty_{j=1} \sum^\infty_{i=0} (i + j + 1) \cdot \Big| H^{(1)}_{i,j} M^{(1)}_1 - H^{(2)}_{i,j} M^{(2)}_1 + \sum^i_{k=0} \frac{1}{\nu + 1} \Big( \frac{\nu}{\nu+1} \Big)^{i-k} \Big\{ H^{(1)}_{k,j-1} M^{(1)}_1 - H^{(2)}_{k,j-1} M^{(2)}_1 \Big\} \Big|\\
    \leq & \big[  2 \beta q v + \beta p ] \cdot z \cdot \big| \big|  (\mathbf{H}^{(1)}, \mathbf{M}^{(1)}) - (\mathbf{H}^{(2)}, \mathbf{M}^{(2)})  \big| \big|_\eta +  \beta p \sum^\infty_{j=1} \sum^\infty_{k=0} \big( \nu + j + k+1 \big) \cdot \Big| H^{(1)}_{k,j-1} M_1^{(1)} - H^{(2)}_{k,j-1} M^{(2)}_1 \Big| \\
    \leq & \big[ \beta p ( \nu + 3) + 2 \beta q v z \big] \cdot z \cdot \big| \big|  (\mathbf{H}^{(1)}, \mathbf{M}^{(1)}) - (\mathbf{H}^{(2)}, \mathbf{M}^{(2)})  \big| \big|_\eta
\end{align*}
where we have interchanged the order of summation and used the geometric summation. Therefore, 
\begin{align*}
    \sup_{\substack{(\mathbf{H}^{(1)}, \mathbf{M}^{(1)}) \neq (\mathbf{H}^{(2)}, \mathbf{M}^{(2)}) \\ F (( (\mathbf{H}^{(1,2)}, \mathbf{M}^{(1,2)}) )) ||_\eta \leq z}} \frac{\big| \big| F \big( (\mathbf{H}^{(1)}, \mathbf{M}^{(1)}) \big) - F \big( (\mathbf{H}^{(2)}, \mathbf{M}^{(2)}) \big) \big|\big|_\eta}{\big| \big|  (\mathbf{H}^{(1)}, \mathbf{M}^{(1)}) - (\mathbf{H}^{(2)}, \mathbf{M}^{(2)})  \big| \big|_\eta} \leq \big[ \beta p (\nu + 3) + 2 \beta q v \big] z
\end{align*}
whereby Assumption 4.1 of \textcite{barbour2012law} is satisfied.\\

Using Theorem 1.4, Chapter 6 of \textcite{pazy2012semigroups}, \textcite{barbour2012law} then conclude that the semilinear problem (\ref{vivax_semilinear}) with initial condition  $(\mathbf{H}(0), \mathbf{M}(0)) \in \mathcal{R}_\eta$ has a unique mild solution
\begin{align}
    (\mathbf{H}(t), \mathbf{M}(t)) = R(t) (\mathbf{H}(0), \mathbf{M}(0)) + \int^t_0 R(t-s) F \big( (\mathbf{H}(s), \mathbf{M}(s)) \big) ds \label{contrans_mild_sol}
\end{align}
on the interval $[0, t_\text{max})$, $t_\text{max} > 0$, with
\begin{align*}
    t_\text{max} < \infty \implies \big| \big|  (\mathbf{H}^{(1)}, \mathbf{M}^{(1)}) - (\mathbf{H}^{(2)}, \mathbf{M}^{(2)})  \big| \big|_\eta \to \infty \text{ as } t \to t_\text{max}.
\end{align*}

To establish a bound on the $\eta$-norm of $(\mathbf{H}(t), \mathbf{M}(t))$, we draw on Section \ref{sec::lln_ode}, recognising Equation (\ref{vivax_hybrid_dhk}) as the Kolmogorov forward differential equation for the network of infinite server queues detailed in Section \ref{sec::vivax_queue_network}. Since the FORI --- or equivalently, the infective bite rate experienced by each human ---  is bounded above by $a := \beta p v$, using the generating function given by Equation (\ref{vivax_H_queue}), we see that for all $t \geq 0$
\begin{align*}
    || (\mathbf{H}(t), \mathbf{M}(t)) ||_{\eta} &= M_0 + M_1 + \sum^\infty_{i=0} \sum^\infty_{j=0} (i + j + 1)  H_{i,j}(t)\\
    &= 2 + \frac{1}{v} \bigg[ \frac{\partial H}{\partial x} (1, 1, t) +  \frac{\partial H}{\partial y} (1, 1, t) \bigg]_{\{ I_M(\tau) = 1 \}} \\
    & = 2 + \beta p \int^{t}_0  \Big( \nu  [ p_H(t - \tau) + p_A(t - \tau) \big] + e^{-\gamma(t - \tau)} \Big) d \tau\\
    & \leq 2 + \beta p \Big[ \frac{\nu (\alpha + \gamma)}{\gamma(\alpha + \mu)} \Big( \frac{\alpha}{\gamma} + \frac{\mu - \gamma}{\alpha + \mu} \Big) + \frac{1}{\gamma} \Big].
\end{align*}
Hence, 
\begin{align*}
    \lim_{t \to \infty} || (\mathbf{H}(t), \mathbf{M}(t)) ||_{\eta} < \infty \implies t_\text{max} = \infty.
\end{align*}

\subsection{Assumption 4.2.2 of \textcite{barbour2012law}} \label{appendix::barbour_4.2}

We now verify Assumption 4.2.2 of \textcite{barbour2012law}, reproduced below with modifications in notation/exposition as appropriate.

\begin{assumption} \label{assumption_4.2}
There exists $\zeta \in \mathcal{R}$ with $\zeta^{(h)}_{i,j}, \zeta^{(m)}_k \geq 1$ such that
    \begin{align*}
        Y := & \sum_{J \in \mathcal{J}} \omega_J \Big( |\mathbf{M}|_1^{-1} (\mathbf{H}, \mathbf{M} ) \Big)  \cdot \Bigg\{ \sum_{k=0,1} \big| J^{(m)}_{k} \big| \zeta^{(m)}_{k}  +  \sum^\infty_{i=0} \sum^\infty_{j=0} \big| J^{(h)}_{i,j} \big| \zeta^{(h)}_{i,j} \Bigg\}\\
         \leq & \big\{ k_1 |\mathbf{M}|_1^{-1} S_r \big( (\mathbf{H}, \mathbf{M} ) \big) + k_2 \big\}^b
    \end{align*}
for some $1 \leq r := r(\zeta) \leq r_\text{max}^{(2)}$ and some $b=b(\zeta) \geq 1$, and that
\begin{align*}
    Z := \sum^\infty_{i=0} \sum^\infty_{j=0} \frac{\eta^{(h)}_{i,j}\big( \big| A_{H_{i,j}, H_{i,j}} \big| + 1 \big) }{\sqrt{\zeta^{(h)}_{i,j}}} + \sum_{k=0,1} \frac{\eta^{(m)}_{k}\big( \big| A_{M_{k}, M_{k}} \big| + 1 \big) }{\sqrt{\zeta^{(m)}_{k}}} < \infty. 
\end{align*}
\end{assumption}

\begin{claim}
Fix $v = |\mathbf{M}|_1/|\mathbf{H}|_1$. Assumption \ref{assumption_4.2} holds for $\zeta = (\zeta^{(h)}, \zeta^{(m)}) \in \mathcal{R}$, such that
\begin{align*}
    \zeta^{(h)}_{i,j} = (i+j+1)^9, \, \zeta_k^{(m)} = 1,
\end{align*}
with $b(\zeta)=1$ and $r(\zeta)=10$.
\end{claim}

\begin{proof}
For our choice of $\zeta \in \mathcal{R}$, we have
\begin{align*}
    Y = & g + \beta q v \sum^\infty_{j=1} \sum^\infty_{i=1} \frac{H_{i,j}}{|\mathbf{M}|_1} + \sum^\infty_{i=0} \sum^\infty_{j=0} ( \gamma j + \mu i \big) \frac{H_{i,j}}{| \mathbf{M}|_1} \big[ (i + j)^9 + (i+j+1)^9 \big] \\
    & + \beta p \frac{M_1}{|\mathbf{M}|_1} \sum^\infty_{i=0} \sum^\infty_{j=0} \frac{H_{i,j}}{|\mathbf{M}|_1} \sum^\infty_{k=0} \frac{1}{\nu + 1} \Big( \frac{\nu}{\nu+1} \Big)^k \cdot \big[ (i + j + k +2)^9 + (i + j + 1)^9 \big]\\
    \leq & g + \beta q +  2 (\gamma + \mu) \sum^\infty_{i=0} \sum^\infty_{j=0} \frac{H_{i,j}}{|\mathbf{M}|_1} (i+j+1)^{10} + 2 \beta p \sum^\infty_{i=0} \sum^\infty_{j=0} \frac{H_{i,j}}{|\mathbf{M}|_1} (i + j + 1)^9 \EX \big[ (H_b + 2)^9 \big]\\
    \leq & g + \beta q + \Big( \gamma + \mu + \beta p \EX \big[ (H_b + 2)^9 \big] \Big) S_{10} \Big( |\mathbf{M}|_1^{-1} (\mathbf{H}, \mathbf{M} ) \Big) 
\end{align*}
where $H_b$ is a geometrically-distributed random variable with mean $\nu$ and state space $\mathbbm{Z}_{\geq 0}$. Further,
\begin{align*}
    Z & :=  1 + g + \sum^\infty_{i=0}  \sum^\infty_{j=0} (i+j+1)^{-7/2} \big( 1 + (\alpha + \mu) i + j \gamma \big) \\
    & \leq 1 + g + (\gamma + \alpha + \mu + 1) \sum^\infty_{i=0} \sum^\infty_{j=0} (i+j+1)^{-5/2} < \infty,
\end{align*}
so Assumption \ref{assumption_4.2} is satisfied.
\end{proof}

\newpage

\section{Proof of Theorem \ref{theorem::ide_ss}} \label{appendix:theorem_ide_ss}
\begin{proof}
Denote by $I_M^*$ the steady state FORI. For notational convenience, we set
\begin{align*}
    y(\tau) = \frac{e^{-\gamma \tau} + \nu p_A(\tau)}{1 + \nu p_A(\tau)}.
\end{align*}
and note that
\begin{align*}
    c_1:= \int^\infty_0 y(\tau) d \tau < \infty \qquad \qquad c_2:= \int^\infty_0 \tau y(\tau) d \tau < \infty.
\end{align*}

From the IDE (\ref{ide_im}), we deduce that each equilibrium solution $I_M^*$ necessarily satisfies
\begin{align}
    \frac{g I_M^*}{b(1-I_M^*)} = 1 - e^{-ac_1 I_M^*}. \label{eq:ss_im}
\end{align}
The disease free equilibrium $I_M^* = 0$ clearly satisfies Equation (\ref{eq:ss_im}). We begin by characterising the existence of endemic equilibria, following analogous reasoning to \textcite{mehra2022hybrid}. On the domain $I_M^* \in [0, 1]$, we observe that:
\begin{itemize}
    \item $f_1(I_M^*) := \frac{g I_M^*}{b(1-I_M^*)}$ is monotonically increasing, convex with $f_1(I_M^*) \to \infty$ as $I_M^* \to 1$
    \item $f_2(I_M^*) := 1 - e^{-ac_1 I_M^*}$ is monotonically increasing, concave.
\end{itemize}
Define $F(I_M^*):= f_1'(I_M^*)- f_2'(I_M^*)$. Then $F'(I_M^*) > 0$ for all $I_M^* \in [0, 1]$, so $F(I_M^*)$ is monotonically increasing on the interval $I_M^* \in [0, 1]$ with $F(I_M^*) \to \infty$ as $I_M^* \to 1$. We thus consider two cases:
\begin{itemize}
    \item \textbf{Case 1}: $F(0) > 0$\\
    Here, it follows that $F(I_M^*) > 0$ for all $I_M^* \in [0, 1] \implies (f_1-f_2)$ is monotonically increasing and thus non-zero for all $I_M^* \in (0, 1]$.
    \item \textbf{Case 2}: $F(0) < 0$\\
    By the intermediate value theorem, there exists $M_0 \in (0, 1)$ such that $F(I_M^*)<0$ for all $I_M^* \in [0, M_0)$ and $F(I_M^*)>0$ for all $I_M^* \in (M_0, 1]$. So $(f_1-f_2)$ is monotonically decreasing on the interval $[0, M_0)$ and monotonically increasing on the interval $(M_0, 1)$. Since $(f_1-f_2)(0)=0$ and $(f_1-f_2)(I_M^*) \to \infty$ as $I_M^* \to 1$, it follows there exists unique $M_1 \in (M_0, 1)$ such that $(f_1-f_2)(M_1)=0$.  
\end{itemize}
Therefore, an endemic equilibrium solution exists iff
\begin{align*}
    F(0) < 0 \iff R_0 := \sqrt{\frac{abc_1}{g}} > 1,
\end{align*}
and given $R_0>1$, this endemic equilibrium is unique.\\

We now characterise the asymptotic stability of the disease-free and endemic equilibria. By Theorem 2 of \textcite{brauer1978asymptotic}, an equilibrium solution $I_M(t) = I_M^*$ is uniformly asymptotically stable if the trivial solution $u(t)=0$ to the first-order linear IDE
\begin{align*}
    \frac{d u}{dt} = - \big[ g + b \big (1-e^{-a c_1 I_M^*} \big) \big] \cdot u(t) + ab (1-I_M^*) e^{-a c_1 I_M^*} \int^t_0 u(t-\tau) P(\tau) d \tau
\end{align*}
is uniformly asymptotically stable. By Theorem 1 of \textcite{brauer1978asymptotic}, since $P(\tau)$ is continuous and non-negative for all $\tau \in [0, \infty)$, it follows that $u(t) = 0$ is uniformly asymptotically stable if and only if
\begin{align}
    - \big[ g + b \big (1-e^{-a c_1 I_M^*} \big) \big] + abc_1 (1-I_M^*) e^{-a c_1 I_M^*} < 0, \label{eq:im_stability}
\end{align}

Equation (\ref{eq:im_stability}) serves as a sufficient condition for a steady state solution $I_M^*$ to be uniformly asymptotically-stable. Noting that any steady state solution $I_M^*$ must solve Equation (\ref{eq:ss_im}), we derive the equivalent condition
\begin{align}
    h(I_M^*) :=  (I_M^*)^2 - \Big( 1 + \frac{b}{b + g} \Big) I_M^* + \frac{b}{b + g} \Big( 1 - \frac{1}{R_0^2} \Big) < 0. \label{eq:im_stability_2}
\end{align}

Using Equation (\ref{eq:im_stability_2}), it immediately follows that the disease-free equilibrium $I_M^* = 0$ is uniformly asymptotically stable in the case where $R_0 < 1$.\\

When $R_0>1$, we have $h(0)>0$ and $h(1)<0$, so there exists precisely one root $h(x) =0$ in the domain $x \in (0, 1)$ and
\begin{align*}
    h(I_M^*) < 0 \iff I_M^* > \frac{1 + \frac{b}{b+g} - \sqrt{\big( 1 - \frac{b}{b+g} \big)^2 + 4 \frac{b}{b+g} \frac{1}{R_0^2} }}{2}
\end{align*}
serves as a \textit{sufficient} condition for the endemic equilibrium to be uniformly asymptotically stable.

\end{proof}

\newpage

\section{The pseudeoquilibrium approximation to the superinfection recovery rate proposed by \textcite{white2014modelling}} \label{appendix::pseudoeq}

\textcite{white2014modelling} construct a deterministic compartmental model, consisting of an infinite-dimensional system of ODEs monitoring the hypnozoite burden and the absence or presence of blood-stage infection. Superinfection is captured through a ``pseudoequilibrium approximation'' for the rate of recovery from blood-stage infection, conditional on the size of the hypnozoite reservoir. Below, we characterise this construction in detail.\\

The pseudoequilibrium approximation to the rate of recovery from (super)infection was initially proposed by \textcite{dietz1974malaria}. For an $M/M/\infty$ queue with homogeneous arrival rate $\lambda$ and service rate $\gamma$ --- which describes within-host superinfection dynamics for \textit{P. falciparum} in a constant transmission setting --- \textcite{dietz1974malaria} note that the expected duration of each busy period, or `superinfection' at steady state is
\begin{align*}
    \EX[ \text{superinfection} ] = \frac{e^{\frac{\lambda}{\gamma}}-1}{\lambda}.
\end{align*}

Therefore, as a function of the FORI $\lambda(t) = \lambda$ at time $t$, \textcite{dietz1974malaria} argue that the rate of recovery from superinfection at time $t$ can be approximated to be
\begin{align}
    \rho(\lambda) = \frac{\lambda}{e^{\frac{\lambda}{\gamma}}-1}. \label{dietz_pseudoeq_approx}
\end{align}

This argument essentially replaces a rate of recovery, which necessarily depends on the MOB, with its expectation taken under steady-state conditions with the arrival rate taken to be constant. While this might be a reasonable thing to do if there is no better option, as we show in this paper, the model is still tractable when the MOB is included in the state space. It might be that the pseudoequilibrium approximation yields reasonable results when changes in the FORI occur on a slower time scale than the clearance of blood-stage infection; it has, however, been critiqued on the basis of its inability to capture MOB dynamics following abrupt changes in transmission intensity \parencite{smith2009endemicity}.\\

To account for superinfection, \textcite{white2014modelling} suggest that the recovery rate from blood-stage infection, conditional on the hypnozoite burden $i$ should be
\begin{align}
    \rho_i( \lambda ) = \frac{\lambda + i \alpha}{e^\frac{\lambda + i \alpha}{\gamma} + 1} \label{white_superinf}
\end{align}
where $\lambda$ represents the FORI. For an individual with a hypnozoite reservoir of size $i$, additional blood-stage infections are acquired at rate $\lambda + i \alpha$ --- with primary infections occurring at rate $\lambda$ due to mosquito inoculation, and each hypnozoite in the liver activating independently at rate $\alpha$ to give rise to a relapse. Equation (\ref{white_superinf}) is a modification of Equation (\ref{dietz_pseudoeq_approx}) accounting for the additional contribution of hypnozoite activation to the blood-stage infection burden. However, like the approximation (\ref{dietz_pseudoeq_approx}) it still arises from replacing a recovery rate that depends on the current value of the MOB with its expectation under steady-state conditions, in which case the MOB would have a Poisson distribution with parameter $(\lambda + i \alpha)/\gamma$.\\

If a single hypnozoite were established through each infective bite, then the distribution of MOB conditional on the size of the hypnozoite reservoir would be Poisson-distributed; for a network of infinite server queues with single arrivals, the numbers of busy servers in each queue follow independent Poisson distributions at each time $t$ \parencite{massey1993networks}. The batch structure in hypnozoites established through each infective bite, however, breaks this construction; indeed, \textcite{boucherie1993transient} show that the only open networks of queues that have transient product form distributions are networks of $M_t/M/\infty$ queues.\\

The distribution of MOB, conditional on the size of the hypnozoite reservoir, can be recovered using the within-host queueing model introduced in \textcite{mehra2022hypnozoite} (see Section \ref{sec::vivax_queue_network} for a summary). Suppose the time evolution of the FORI $\lambda(t)$ is known. Recall that $N_H(t)$ represents the size of the hypnozoite reservoir at time $t$ and let $N_B(t) = N_A(t) + N_P(t)$ denote the MOB at time $t$. By \textcite{xekalaki1987method}, it follows that
\begin{align*}
    \EX \big[ y^{N_B(t)} | N_H(t) = i \big] &= \frac{\partial^n H}{\partial x^n} (1, y, t) \Bigg/ \frac{\partial^n H}{\partial x^n} (0, y, t)
\end{align*}
where $H$ denotes the generating function given by Equation (\ref{vivax_H_queue}). Following similar reasoning to \textcite{nasell2013hybrid}, given a hypnozoite reservoir of size $i$, the rate of recovery from (super)infection can be written
\begin{align*}
    \rho_i(t) = \gamma \frac{P(N_B(t) = 1 | N_H(t) = i)}{P(N_B(t) > 0 | N_H(t) = i)}.
\end{align*}

Explicit formulae for the rate of recovery from superinfection, in terms of partial Bell polynomials, can be recovered using Leibniz's integral rule and Faa di Bruno's formula \parencite{xu2011unified}; the resultant expressions, however, are unwieldy and are therefore omitted here. A simpler framework -- formulated in terms of a recovery rate corrected for superinfection -- is detailed in Appendix \ref{appendix::superinf_ide}.

\newpage

\section{A correction to the model of \textcite{anwar2021multiscale}} \label{appendix::superinf_ide}

\textcite{anwar2021multiscale} construct a multiscale transmission model for \textit{P. vivax}, allowing for explicit variation in the hypnozoite burden. The objective of \parencite{anwar2021multiscale} is to capture the model structure proposed by \textcite{white2014modelling}, whilst alleviating the computational difficulties associated with infinite-dimensional compartmental models. The premise of \parencite{anwar2021multiscale} is to embed the probabilistic distributions derived on the within-host scale in \textcite{mehra2022hypnozoite} in a transmission model with finitely-many compartments, with the human population stratified by the presence or absence of liver- and/or blood-stage infection.\\

In particular, as a function of the FORI $\lambda(t)$, \parencite{anwar2021multiscale} derives a system of IDEs governing the time evolution of the prevalence of blood-stage infection $I_H(t)$ in a homogeneous population. To do so, \parencite{anwar2021multiscale} draws on the relapse rate, conditional on the absence of blood-stage infection, that we derive at the within-host level in \parencite{mehra2022hypnozoite}.\\

This within-host framework is designed to capture the joint dynamics of superinfection and the hypnozoite reservoir. As a function of the FORI $\lambda(\tau)$ in the interval $\tau \in [0, t)$, using Equation (77) of \parencite{mehra2022hypnozoite}, the relapse rate $r(t)$ conditional on the absence of blood-stage infection at time $t$ can be written 
\begin{align*}
    r(t) := \alpha \EX[ N_H(t) | N_A(t) + N_P(t) = 0] = \nu \alpha \int^t_0 \frac{\lambda(\tau) p_H(t-\tau) \big( 1 - e^{-\gamma (t-\tau)} \big)}{[1 + \nu p_A(t-\tau)]^2} d \tau.
\end{align*}

While this relapse rate is derived under a framework that takes superinfection into account, the population-level framework proposed in \parencite{anwar2021multiscale} does not track the proportion of individuals in the population that have specific values of the MOB, that is it does not account for superinfection in modelling the recovery rate of individuals. In fact, all blood-stage infected individuals are assumed to recover from infection at a constant rate $\gamma$. The inclusion of superinfection in the work of \parencite{anwar2021multiscale} would necessitate the derivation of the recovery rate from (super)infection under the within-host framework of \parencite{mehra2022hypnozoite}. As in \textcite{nasell2013hybrid}, the rate of recovery from (super)infection could be written in terms of the within-host PMF for the MOB $N_A(t) + N_P(t)$ at time $t$ (given by Equations (81) and (82) of \parencite{mehra2022hypnozoite}):
\begin{align}
    \rho(t) = \frac{\gamma P(N_A(t) + N_P(t) = 1)}{P(N_A(t) + N_P(t) >0 )} = \frac{\gamma \int^t_0 \lambda(\tau) \frac{e^{-\gamma(t-\tau)} + \nu p_A(t-\tau)}{[1 + \nu p_A(t-\tau) ]^2} d \tau}{\exp \Big\{ \int^t_0 \lambda(\tau) \frac{e^{-\gamma(t-\tau)} + \nu p_A(t-\tau)}{1 + \nu p_A(t-\tau) } d \tau \Big\} - 1},
\end{align}
which, in turn, is a function of the FORI $\lambda(\tau)$ in the interval $\tau \in [0, t)$. A rigorous characterisation of superinfection would yield the ODE
\begin{align*}
    \frac{d I_H}{dt} = \big( r(t) + \lambda(t) \big) (1 - I_H) - \rho(t) I_H.
\end{align*}
rather than
\begin{align*}
    \frac{d I_H}{dt} = \big( r(t) + \lambda(t) \big)  (1 - I_H) - \gamma I_H,
\end{align*}
which appeared in \parencite{anwar2021multiscale}.\\

However, rather than formulating a model in terms of the relapse rate and the rate of recovery from (super)infection, a simpler but equivalent approach is to directly consider the (integral) expression for the expected prevalence of blood-stage infection in the human population over time under the within-host framework of \parencite{mehra2022hypnozoite}, as we have done in this paper.\\

While the proposed correction to the framework of \parencite{anwar2021multiscale} is somewhat subtle, it has important implications. As discussed in \parencite{anwar2021multiscale}, the omission of superinfection at the population-level leads to discrepancies relative to the framework of \parencite{white2014modelling}, the re-formulation of which was the primary aim; these discrepancies are evident at high transmission intensities, where superinfections would be expected. The deterministic model introduced in Section \ref{sec::vivax_hybrid}, in contrast, serves as an extension of the framework of \parencite{white2014modelling} to allow for superinfection. The omission of superinfection at the population-level in \parencite{anwar2021multiscale} means that a number of key epidemiological parameters that are analytically tractable at the within-host level do not make sense at the population-level. A key example is the expected contribution of relapse vs primary infection to the overall burden of recurrence (which includes the case of overlapping relapses and primary infections, as in \parencite{mehra2022hypnozoite}), 

\newpage

\section{The model of \textcite{bailey1957}} \label{sec::bailey}
In the case $\nu=0$ (that is, the absence of hypnozoite accrual), Equations (\ref{vivax_hybrid_dhk}) and (\ref{vivax_hybrid_dim}) reduce to the classical (deterministic) model of superinfection for \textit{P. falciparum} formulated by \textcite{bailey1957}.\\ 

In the case $\nu=0$, Equation (\ref{vivax_hybrid_dhk}) yields the Kolmogorov forward differential equations for the $M_t/M/\infty$ queue. By \textcite{eick1993physics}, it follows that the MOB in the human population at time $t$ necessarily has a Poisson distribution. The mean MOB
\begin{align*}
    m_b(t) := a \int^t_0 I_M(\tau) e^{-\gamma(t-\tau)} d \tau = -\log(1 - I_H(t))
\end{align*}
can be written a function of the FORI $aI_M(\tau)$ in the interval $[0, t)$, with the final equality following from Equation (\ref{eq:prev}).\\

In this case, we can collapse the infinite-dimensional system of ODEs given by Equations (\ref{vivax_hybrid_dhk}) and (\ref{vivax_hybrid_dim}) into the two-dimensional system
\begin{align}
    \frac{dI_H}{dt} &= -\big[ a I_M - \gamma \log(1-I_H) \big] (1 - I_H) \label{eq:bailey_ih}\\
    \frac{dI_M}{dt} &= -gI_M + b I_H(1-I_M) \label{eq:bailey_im}
\end{align}
by differentiating Equation (\ref{eq:prev}). When the MOB in the human population is initially Poisson-distributed, the time evolution of FORI is equivalent under the two-dimensional system of ODEs given by Equations (\ref{eq:bailey_ih}) and (\ref{eq:bailey_im}), and the countable system of ODEs proposed by \textcite{bailey1957} (Equations (\ref{vivax_hybrid_dhk}) and (\ref{vivax_hybrid_dim}) in the case $\nu=0$).\\

Noting the functional dependence between the mean MOB $m_b(t)$ and the prevalence of blood-stage infection $I_H(t)$ at time $t$, we observe that the system of ODEs given by Equations (\ref{eq:bailey_ih}) and (\ref{eq:bailey_im}) is equivalent to the ``hybrid approximation'' proposed by \textcite{naasell1991quasi}, whereby host and vector dynamics are coupled through expected values.\\

Applying the next generation matrix method \parencite{van2002reproduction} to the system of ODEs given by Equations (\ref{eq:bailey_ih}) and (\ref{eq:bailey_im}) yields the basic reproduction number
\begin{align*}
    R_0 = \sqrt{\frac{ab}{g \gamma}},
\end{align*}
which is the same as Equation (\ref{eq:R0}) in the special case $\nu=0$.

\newpage

\printbibliography

\end{document}